\documentclass[10pt]{article}
\usepackage{color, amssymb, amsthm, amsmath, amsfonts, ascmac, comment, enumerate}

\usepackage{graphicx}
\usepackage{hyperref}

\usepackage{multirow}
\usepackage{tcolorbox}
\usepackage{xcolor}
\hypersetup{
  colorlinks, 
  citecolor=blue!20!black!30!green,
  linkcolor=red,
  urlcolor=blue}
\newtheorem{Thm}{Theorem}
\newtheorem*{Thm*}{Theorem}
\newtheorem{Prop}{Proposition}	
\newtheorem{Lem}{Lemma}
\newtheorem{Cor}{Corollary}
\newtheorem{Fact}{Fact}

\newtheorem*{Rem*}{Remark}
\newtheorem*{Thm1}{\rm\bf Theorem~\ref{Thm_unified_amortized}}
\newtheorem*{Thm2}{\rm\bf Theorem~\ref{Thm_unified_amortized_Theta}}
\newtheorem*{Thm3}{\rm\bf Theorem~\ref{Thm_direct_sum_expec_0}}

\usepackage[top=30truemm,bottom=30truemm,left=20truemm,right=20truemm]{geometry}

\newcommand{\Bset}{\{0, 1\}}

\newcommand{\ED}{\overline{D}}
\newcommand{\EC}{\overline{C}}
\newcommand{\ER}{\overline{R}}
\newcommand{\EQD}{\overline{QD}}
\newcommand{\EQR}{\overline{QR}}
\newcommand{\E}{\mathbf{E}}
\newcommand{\on}{\text{~on~}}
\newcommand{\Mand}{\text{~and~}}
\newcommand{\st}{\text{~s.t.~}}
\newcommand{\Tdelta}{\tilde{\delta}}
\newcommand{\TU}{\tilde{\mathcal{U}}}
\newcommand{\FE}{[P_R, \varepsilon]}

\newcommand{\PD}{(\mathcal{F}_\Theta, \NT, \mu)}
\newcommand{\FTE}{[(\mathcal{F}_\Theta, \NT, \mu), \varepsilon]}
\newcommand{\FTEn}{\FTE^{n}}
\newcommand{\PDE}{[P_D, \varepsilon]}
\newcommand{\PDEn}{[P_D, \varepsilon]^n}

\newcommand{\PR}{(\mathcal{F}_\Theta, \NT)}
\newcommand{\PRE}{[P_R, \varepsilon]}
\newcommand{\PREn}{[P_R, \varepsilon]^n}
\newcommand{\PT}{\mathcal{P}(\Theta)}

\newcommand{\PQD}{(f, \NT, \mu)}
\newcommand{\FQTE}{[(f, \NT, \mu), \varepsilon]}
\newcommand{\FQTEn}{\FTE^{n}}
\newcommand{\PQDE}{[P_{QD}, \varepsilon]}
\newcommand{\PQDEn}{[P_{QD}, \varepsilon]^n}

\newcommand{\PQR}{(f, \NT)}
\newcommand{\FQRE}{[(f, \NT), \varepsilon]}
\newcommand{\FQREn}{\FQRE^{n}}
\newcommand{\PQRE}{[P_{QR}, \varepsilon]}
\newcommand{\PQREn}{[P_{QR}, \varepsilon]^n}

\newcommand{\NT}{\mathcal{N}_\Theta}
\newcommand{\Nt}{\mathcal{N}_\theta}

\title{Direct sum theorems beyond query complexity}
\author{Daiki Suruga \thanks{Graduate School of Mathematics, Nagoya University}}
\begin{document}
\maketitle
\begin{abstract}
A fundamental question in computer science is: \emph{Is it harder to solve $n$ instances independently than to solve them simultaneously?} 
This question, known as the direct sum question or direct sum theorem,  has received much attention in several research fields including query complexity, communication complexity and information theory.
Despite its importance, however, little has been discovered in many other research fields.
\par
In this paper, we introduce a novel framework that extends to classical/quantum query complexity, PAC-learning for machine learning, statistical estimation theory, and more. 
Within this framework, we establish several fundamental direct sum theorems.
The main contributions of this paper include:
(i) establishing a complete characterization of the amortized query/oracle complexities, and
(ii) proving tight direct sum theorems when the error is small.
Note that in our framework, every oracle access needs to be performed \emph{classically}, even though our framework is capable of both classical and quantum scenarios.
This can be thought of one limitation of this work.
\par
As a direct consequence of our results, we obtain:
\begin{itemize}
\item The first known asymptotic separation of the randomized query complexity. Specifically, we show that there is a function $f: \Bset^k \to \Bset$ and small error $\varepsilon > 0$ such that solving $n$ instances simultaneously requires the query complexity $\tilde{O}(n\sqrt{k})$ but solving one instance with the same error has the complexity $\tilde{\Omega}(k)$.
In communication complexity this type of separation was previously given in~Feder, Kushilevitz, Naor and Nisan (1995).
\item The query complexity counterpart of the ``information = amortized communication" relation, one of the most influential results in communication complexity shown by Braverman and Rao (2011) and further investigated by Braverman (2015).
\item A partial answer to an open question given in~Jain, Klauck and Santha (2010), by showing a tighter direct sum theorem.
\item A complete answer to the open problem given in~Blais and Brody (2019) by exhibiting a counterexample.
\end{itemize}
We hope that our results will provide further interesting applications in the future.
\end{abstract}

\section{Introduction}
The \emph{direct sum question} is a basic, natural and fundamental question in complexity theory which asks whether it is easier to solve $k$ instances of a problem simultaneously than to solve each of them independently.
This question and its variants (e.g., XOR lemmas, direct product theorems) have attracted much attention in several research fields such as query complexity~\cite{JKS10, ACLT10, Dru11, Mon13, MS15, BK18, BB19, BB20, GM21, BTLS23, BKST24}, communication complexity~\cite{FKNN95, CSWY01, JRS03d, BJKS04, BBCR09, JK09, BR11, JPY12, MWY13, KMSY14, Bra17, Jain20, JK22, Wu22}, Boolean circuits~\cite{IW97, IJKW08, GNW11}.
As a consequence of these efforts, it is now known that the direct sum theorems hold in some models~\cite{JK09, JKS10, ACLT10, BB19} such as the deterministic query algorithm, 
whereas such theorems do not hold in several other models such as the two-party randomized classical communication model~\cite{FKNN95}.
(See Ref.~\cite[Section~3]{Pan12} for a survey of direct sum theorems.)
Providing various kinds of applications in addition to its original significance, the direct sum theorems have been playing a key role in complexity theory.

\subsection{Direct sum theorems in query and communication complexity}\label{subsec_DST-in-QC}
In this section, a brief review of the direct sum theorems in query and communication complexity is given, focusing especially on results relevant to the present paper.
\paragraph{Direct sum in query complexity}
In classical query complexity, several basic properties on the direct sum theorems are proved in~\cite{JKS10}, which shows
\begin{equation}\label{eq_intro_JKS10}
\mathsf{Det}(f^n) = n\mathsf{Det}(f)
\Mand R([f^n, \varepsilon]) \geq \delta^2 n   R([f, \varepsilon/(1-\delta) +\delta])
\end{equation}
where $\mathsf{Det}(f)$ denotes the deterministic query complexity for computing $f$, and $R([f, \varepsilon])$ denotes the worst-case randomized query complexity for computing $f$ with the worst-case error $\leq \varepsilon$, and $f^n = \underbrace{(f, \ldots, f)}_n$.
 Ref.~\cite{BK18} then showed
\begin{equation*}
\ER([f^n, \varepsilon]) \geq n\ER([f, \varepsilon])
\end{equation*}
holds where $\ER([f, \varepsilon])$ denotes the \emph{expected} randomized query complexity for computing $f$ with the worst-case error $\leq \varepsilon$.
This result is then strengthened by~\cite{BB19}, which firstly characterize the tight direct sum theorem as

\begin{equation}\label{Intro_eq_direct-sum_randomized}
\ER([f^n, \varepsilon]) = \Theta(n\ER([f, \varepsilon/n])).
\end{equation}
These results show that the direct sum theorems hold in the worst-case/expected randomized query complexity.
\par
Unlike the randomized model, it is well-known that in general, direct sum theorems do not hold in the worst-case distributional query complexity~\cite{Sha01}.
As Ref.~\cite{Sha01} shows\footnote{Precisely speaking, they showed 
$D([f^{ n}, \mu^n, \varepsilon]) = O\left(\varepsilon D([f, \mu, \varepsilon/n]) +n \log \frac{n}{\varepsilon}\right)$
but the term $n \log \frac{n}{\varepsilon}$ is independent of the function $f$ and therefore ignored.
}, there is a function $f$ such that
\begin{equation*}
D([f^{ n}, \mu^n, \varepsilon])
= O\left(\varepsilon D([f, \mu, \varepsilon/n])\right)
\end{equation*}
holds where $D([f, \mu,  \varepsilon])$ denotes the worst-case query complexity for computing $f$ with the average error under the distribution $\mu$. Since the RHS is trivially upper-bounded by $\lceil \log |\mathrm{dom} f|\rceil$, the LHS can not grow arbitrarily larger even if $n$ gets larger.
On the other hand, recently in Ref.~\cite{BKST24}, the authors showed the direct sum theorem \emph{does} hold in the \emph{expected} distributional query complexity:
\begin{equation}\label{eq_intro_BKST}
\ED([f^n, \mu^n, \varepsilon])
= \tilde{\Omega}(\varepsilon^2 n)\ED([f, \mu, \Theta(\varepsilon/n)])
\end{equation}
where $\ED([f, \varepsilon])$ denotes the \emph{expected} distributional query complexity for computing $f$ with the average error $\leq \varepsilon$ under the distribution $\mu$.

\par
Similar to the classical case, there are plentiful amount of researches in \emph{quantum} query complexity, including the groundbreaking Grover's search algorithm~\cite{Gro96}.
Regarding the direct sum question, a tight characterization on the worst-case quantum query complexity has been firstly shown in 2010 by~\cite{ACLT10};
Ref.~\cite{ACLT10} shows
\begin{equation}\label{Intro_eq_direct-sum_quantum}
QR([f^n, 1/3])
= \Theta(n QR([f, 1/3]))
\end{equation}
where $QR([f, \varepsilon])$ denotes the worst-case randomized quantum query complexity for computing $f$ with the worst-case error $\leq \varepsilon$.
The direct sum question in quantum query complexity is further investigated in several works~\cite{LMR+11}.

\par
These line of researches guarantees the importance of the direct sum questions in query complexity, even though historically it was sometimes mistakenly regarded as unimportant.
(See~\cite[Introduction]{JKS10} for a discussion.)

\paragraph{Direct sum in communication complexity}
Communication complexity definitely plays a central role in complexity theory~\cite{KN96, RY20}.
In communication complexity, several fundamental properties have been firstly proved in~\cite{FKNN95}.
Ref.~\cite{FKNN95} shows (among other results) there is a function $f:\Bset^\ell \to \Bset$ that satisfies
\begin{equation}\label{Intro_eq_FKNN}
R^\mathrm{CC}([f^n, 1/3]) = \Theta(n)
\Mand
R^\mathrm{CC}([f, 1/3]) = \Theta(\log \ell)
\end{equation}
where
$R^\mathrm{CC}([f, \varepsilon])$ denotes the worst-case randomized communication complexity for computing $f$ with the worst-case error $\leq \varepsilon$.
This result means, in the randomized communication complexity, the direct sum theorem simply does not hold.
Also note that as complementary results, it is shown in~\cite{JRS03d, JK09} that  the direct sum theorems hold when focusing on the restricted model of communication, e.g., the simultaneous message model~\cite{JK09}.

\par
One of the most essential tool for the analysis of the communication complexity is \emph{information complexity}, introduced originally in~\cite{CSWY01} for the simultaneous message model and relatively recently in~\cite{BR11} for the general two-party model.
There are a considerable number of works that apply the information complexity framework to the direct sum theorems in communication complexity.
This is partly because the quantity called \emph{information complexity}, which characterizes how much information the two parties need to reveal (see~\cite{BR11} for the precise definition), itself satisfies some version of the direct sum theorems~\cite{BR11, Bra15}.
Applying the information complexity framework,
Ref.~\cite{BR11} shows the complete characterization of the amortized two-party communication complexity in the distributional setting:
\begin{equation}\label{Intro_eq_info-comp_dist}
\lim_{n \to \infty} \frac{D^\mathrm{CC}([f, \mu, \varepsilon]^n)}{n}
= IC([f, \mu, \varepsilon])
\end{equation}
where $D^\mathrm{CC}([f, \mu, \varepsilon]^n)$ denotes the distributional communication complexity for computing $f^n$ with error $\leq \varepsilon$ on each  of $n$ instances $f$ under the input distribution $\mu$, and 
$IC([f, \mu, \varepsilon])$ denotes the information complexity for computing $f$ with error $\leq \varepsilon$ under the input distribution $\mu$.
Subsequently, Ref.~\cite{Bra15} applies the information complexity framework and shows a similar result but for the randomized setting, 
whereas Ref.~\cite{Tou15} has generalized the relation~\eqref{Intro_eq_info-comp_dist} to the quantum setting.
The information complexity has undoubtedly become an essential tool for investigating many topics in communication complexity~\cite{BJKS04, BEO+13, TVVW17, BGK+18}, as well as  direct sum theorems~\cite{BBCR09, MWY13, KMSY14, Jain20}.

\par
Similar to query complexity, the direct sum question in communication complexity has been extensively studied for better understanding of communication complexity.

\par
To summarize, as seen in both query complexity and communication complexity, the direct sum theorems are fundamental issues and worth investigating, providing various kinds of applications. 
However, despite the importance, little is discovered in many other complexity frameworks such as statistical sample complexity.
Therefore, it is necessary to investigate direct sum questions in those less-investigated complexity frameworks, as well as to provide more precise analysis in the well-investigated frameworks such as query complexity and communication complexity.

\subsection{Our contributions}
As mentioned in Section~\ref{subsec_DST-in-QC}, there are many research fields that direct sum theorems have not received much attention to, despite its importance.
In this paper, overcoming the issue, we introduce a new general framework that enables to handle various kinds of research topics such as classical/quantum query complexity, statistical estimation problems, PAC learning for machine learning. (See Section~\ref{subsec_examples} for a detailed explanation.)
Under the new framework, we then successfully analyze different research problems in a unified manner and prove several fundamental direct sum theorems applicable to any of these research problems.
\par
In the following sections, we first explain our new framework and then the two sections for our main results follow.
Since our main results many be divided into the two parts:``Direct sum theorems \emph{in the limit}'' and ``Direct sum theorems \emph{without the limit}'', we describe each of the two results separately after the explanation for our new framework.

\paragraph{Our framework}
As our new framework provides a pivotal role in this paper, we now describe its definition a bit in detail. The precise definition is given in Section~\ref{sec_Preliminaries}.
For simplicity, let us focus on that of classical randomized scenarios even though in this paper the new framework is applied to any of classical or quantum, distributional or randomized scenarios.
\par
First recall the well-known framework: the classical randomized query complexity. In the classical randomized query complexity, one needs to compute the value $f(x)$ of a function $f: \Bset^\ell \to \Bset$ by accessing to an oracle (or a query) that takes $i \in [\ell]$ as input and output $x_i$ deterministically.
The key differences between the query complexity framework and our framework are the definitions of  (i)~target functions and (ii)~oracles:

\begin{enumerate}[(i)]
\item
In our framework a target function is denoted as $F_\Theta: \Theta \ni \theta \mapsto F_\theta \subset \mathbb{R}^d$; the domain is simply a (possibly infinite) set $\Theta$
and the output value $F_\theta$ is a subset of some fixed Euclidean space $\mathbb{R}^d$.
By taking $\Theta := \Bset^\ell$ and $F_\theta := \{f(\theta)\}~(\theta \in \Bset^\ell)$ which has only one element $f(\theta)$, we see this definition covers the function in the query complexity framework.
\item
In case of oracles, our new definition allows them to behave stochastically. 
That is, in our framework, an oracle, denoted by $\mathcal{N}_\Theta =\{p_\theta(y|x)\mid \theta \in \Theta\}$, takes $x \in \mathcal{X}$ as input and return $y \in \mathcal{Y}$ with probability $p_\theta(y|x)$, where $\mathcal{X}$ and $\mathcal{Y}$ are finite sets\footnote{In information theory, this definition is known to be equivalent to classical channels.}.
\end{enumerate}
In short, in our framework, any problem $P$ is represented by the pair $(F_\Theta, \mathcal{N}_\Theta)$ whereas in query complexity any problem is defined only by a function $f$.
These are the main differences between the query complexity framework and our new framework.
Within the new framework, the player's mission is to output some real value $\pi_\mathsf{out} \in F_\theta$ by sending $x_1, \ldots, x_m \in \mathcal{X}$ to the oracle and receiving $y_1, \ldots, y_m \in \mathcal{Y}$ from the oracle in an adaptive manner.
Note that in our model of computation, unlike the model given in~\cite{ACLT10, LMR+11}, every oracle access is made in a classically adaptive way
even if we consider quantum information processing.
(See Section~\ref{sec_Preliminaries} for a detailed description on our model of computation.)
As shown in Section~\ref{subsec_examples}, this framework enables to investigate different complexity frameworks in a unified manner.

\paragraph{First part: Direct sum theorems in the limit}
In the first part of our results, we concern with direct sum theorems in the limit under our framework. To state our results in a concise manner, let us introduce several notations in the following. 
For each of the four complexity scenarios—classical distributional, classical randomized, quantum distributional, and quantum randomized—we use the abbreviations $D$, $R$, $QD$, and $QR$, respectively.
Then for any complexity scenario $C \in \{D, R, QD, QR\}$ and any problem $P_C$ (with the subscript $C$ to express which scenario is considered), 
let $C([P_C, \varepsilon])$ (resp. $\EC([P_C, \varepsilon])$) be the worst-case (resp. the expected) oracle complexity of the problem $P_C$ with error $\leq \varepsilon$.
For example, $QR([P_{QR}, 1/3])$ denotes the worst-case oracle complexity of the problem $P_{QR}$ with error $\leq 1/3$.
For direct sum theorems, we also define $C([P_C, \varepsilon]^n)$ (resp. $\EC([P_C, \varepsilon]^n)$) as the worst-case (resp. the expected) oracle complexity of the problem $P_C^n = \underbrace{(P_C, \ldots, P_C)}_n$ with error $\leq \varepsilon$ on \emph{each instance} $P_C$. 
Note that as is already defined,  $C([P_C^n, \varepsilon])$ denotes the complexity of $P_C^n$ with error $\leq \varepsilon$ on \emph{all $n$ instances}, even though $C([P_C^n, \varepsilon])$ and $C([P_C, \varepsilon]^n)$ may look similar.

\par
Using the notations defined above, one of the main results is stated as follows:
\begin{Thm}\label{Thm_unified_amortized}
For any complexity scenario $C \in \{D, R, QD, QR\}$, any $\varepsilon > 0$, and any problem $P_C$,
\begin{equation*}
\lim_{n \to \infty} \frac{C([P_C, \varepsilon]^n)}{n}
=
\EC([P_C, \varepsilon]).
\end{equation*}
\end{Thm}
Theorem~\ref{Thm_unified_amortized} firstly gives a complete characterization of the worst-case complexity $C([P_C, \varepsilon]^n)$ in the asymptotic setting, which had not been discovered before.
In classical scenarios, i.e., $C \in \{D, R\}$, Theorem~\ref{Thm_unified_amortized} naturally corresponds to the query/oracle counterpart of ``information = amortized communication'' relations~\cite{BR11, Bra15} as in the expression~\eqref{Intro_eq_info-comp_dist}, whereas in quantum cases Theorem~\ref{Thm_unified_amortized} arguably does not correspond to that of~\cite{Tou15} due to the classical adaptivity of our model of computation.
Since the ``information = amortized communication'' relations provide a considerable number of applications, Theorem~\ref{Thm_unified_amortized} may provide several important applications as well in the future.

\par
We also consider the case of $C([P_C^n, \varepsilon])$ with small error $\varepsilon$ and obtain Theorem~\ref{Thm_unified_amortized_Theta}:
\begin{Thm}\label{Thm_unified_amortized_Theta}
(informal)
For any complexity scenario $C \in \{D, R, QD, QR\}$ and for almost any problem $P_C$,
\begin{equation*}
\lim_{n \to \infty} \frac{C([P^n_C, \varepsilon])}{n}
=
\Theta\left(\EC([P_C, 0])\right).
\end{equation*}
for any sufficiently small positive $\varepsilon$.
\end{Thm}
Together with the result~\cite{ABB+16} showing the function satisfying $\ER([f, 0]) = \tilde{O}(\sqrt{\mathsf{Det}(f)})$ as well as Proposition~\ref{Prop_appendix_cont}: $R([f, \varepsilon]) = \mathsf{Det}(f)$ for small $\varepsilon > 0$, Theorem~\ref{Thm_unified_amortized_Theta} gives the following corollary:
\begin{Cor}\label{Cor_separation_FKNN}
There is a function $f$ and small (but not too small) $\varepsilon > 0$ such that
\begin{equation*}
R([f^n, \varepsilon]) = \Theta(n \sqrt{\mathsf{Det}(f)}) \Mand R([f, \varepsilon]) = \mathsf{Det}(f)
\end{equation*}
hold.
\end{Cor}
This firstly gives an asymptotic separation of the type~\eqref{Intro_eq_FKNN} in classical randomized query complexity.
On the other hand, for a relatively large error such as $\varepsilon =1/3$, we can not get any non-constant advantage:
\begin{Cor}
For any boolean valued function $f$,
$R([f^n, 1/3]) = \Omega(n \cdot R([f, 1/3]))$
holds.
\end{Cor}
\begin{proof}
By Markov inequality and the success amplification trick, $\ER([f, 1/3]) = \Omega(R([f, 1/3]))$ holds. Combining with Theorem~\ref{Thm_unified_amortized} shows the statement.
\end{proof}
These are the immediate corollaries from Theorem~\ref{Thm_unified_amortized} and Theorem~\ref{Thm_unified_amortized_Theta} in case of classical randomized query complexity. Other possible applications should be discussed in future research.

\paragraph{Second part: Direct sum theorems without the limit}
The main result for the second part is the following:
\begin{Thm}\label{Thm_direct_sum_expec_0}
(informal)
For any complexity scenario $C \in \{D, R, QD, QR\}$ and for almost any problem $P_C$,
\begin{equation*}
\EC([P_C^n, \varepsilon]) = \Theta(n \cdot \EC([P_C, 0]))
\end{equation*}
holds for any sufficiently small positive $\varepsilon$.
\end{Thm}
Let us discuss several related works related to Theorem~\ref{Thm_direct_sum_expec_0}. In classical randomized complexity, Ref.~\cite{JKS10} showed the basic relations~\eqref{eq_intro_JKS10} and posed a question whether it is possible to eliminate the term $\delta^2$ as well as $\varepsilon/(1 - \delta)$ in the error exponent.
By Markov inequality and the success amplification trick showing $\ER([P_R, 0]) = \Omega(\log (1/\delta) R([P_R, \delta]))$, Theorem~\ref{Thm_direct_sum_expec_0} tells neither of them are required when the error is sufficiently small, and hence partly answers the question.
Another related work is Ref.~\cite{BB19} which shows the relation~\eqref{Intro_eq_direct-sum_randomized} in case of classical randomized query complexity.
Compared to the relation~\eqref{Intro_eq_direct-sum_randomized}, Theorem~\ref{Thm_direct_sum_expec_0} provides a better bound although it is applicable only for small $\varepsilon$.
Additionally, our results answer the open problem posed in Ref.~\cite[the sentence after Theorem~2]{BB19}: ``Whether or not $R([f^n, \varepsilon]) = \Omega(nR(f, \varepsilon/n))$ for any $f$ and $\varepsilon$?'' in the negative way, by the counter example given in Corollary~\ref{Cor_separation_FKNN}.
Lastly, we compare Theorem~\ref{Thm_direct_sum_expec_0} with the recent result~\cite{BKST24} that shows the direct sum relation~\eqref{eq_intro_BKST} in case of classical distributional query complexity.
One possible issue of the relation~\eqref{eq_intro_BKST} is that the bound become trivial for small $\varepsilon$, e.g., $\varepsilon =o(\sqrt{n})$.
Theorem~\ref{Thm_direct_sum_expec_0} overcomes this issue and shows the optimal bound when $\varepsilon$ is small.

\subsection{Proof techniques}

The keys for the proof of our results are the two properties that the quantity $\EC([P_C, \varepsilon])$ has: \emph{Additivity} and \emph{Continuity}.
Here we describe its meaning and how to prove them in detail.

\paragraph{Additivity}
The term ``additivity'' is sometimes used in several fields in information science such as Information theory.
In this work, the additivity property denotes the following:
\begin{equation*}
\EC([P_C, \varepsilon]^n)
=
n\cdot\EC([P_C, \varepsilon]).
\end{equation*}
For proof, we basically apply the following basic strategy:
\begin{itemize}
\item To prove $\EC([P_C, \varepsilon]^n) \leq n\cdot\EC([P_C, \varepsilon])$, take an optimal algorithm for $[P_C, \varepsilon]$ and run the algorithm $n$ times for the $n$ instances $P_C^n$.

\item To prove the opposite direction: $\EC([P_C, \varepsilon]^n) \geq n\cdot\EC([P_C, \varepsilon])$, take an optimal algorithm for $[P_C, \varepsilon]^n$ and take $i \in [n]$ uniformly at random. Then use the optimal algorithm to solve only the $i$'th instance $[P_C, \varepsilon]$.
\end{itemize}
This strategy, in turn, successfully yields the correct proof in case of $C \in \{D, QD\}$. 
However, in case of $C \in \{R, QR\}$ some additional technique is in fact necessary, because we need to optimize the algorithms over all inputs.
We therefore prove a version of minimax theorems~\cite{Neu28} as the additional technique
and apply it to prove the additivity in case of $C \in \{R, QR\}$.
Apart from the present work, several versions of minimax theorems are sometimes used in computer science~\cite{BGPW13, Bra15, BB23}.

\paragraph{Continuity}
The continuity literally means the following:
\begin{equation*}
\lim_{\rho \to \varepsilon} \EC([P_C, \rho])
= \EC([P_C, \varepsilon]).
\end{equation*}
Note that such a property does not hold in the case of the worst-case complexity.
A basic strategy for its proof is as follows.
Take two optimal algorithms $\pi$ for $[P_C, \rho]$ and $\pi'$ for $[P_C, \varepsilon/2]$, and run $\pi$ w.p. $1 - p$ and $\pi'$ w.p. $p$ ($p \in (0, 1)$).
When the probability $p$ is appropriately selected, the new algorithm turns out to have an error $\leq \varepsilon$ and has the complexity $\EC([P_C, \rho]) + O(|\rho - \varepsilon|)$.
This is the basic strategy for the proof and in fact works for any scenario $C \in \{D, R, QD, QR\}$ and any $\varepsilon$ except for $\varepsilon = 0$.
In case of $\varepsilon =0$, the proof is done by a different technique, a careful analysis on the output statistics of algorithms for $[P_C, \varepsilon]$.
Similar techniques have previously appeared in~\cite{BR11, BGPW13}.

\subsection{Organization of the paper}
Section~\ref{sec_Preliminaries} describes the notations, our models of computation and examples captured by our framework.
Section~\ref{sec_technical} collects several mathematical assumptions and facts used to prove some of our results.
Section~\ref{sec_additivity} is devoted for the proof of the additivity, and 
Section~\ref{sec_continuity} is done for the proof of the continuity.
Section~\ref{sec_opt} describes constructions of optimal algorithms.
Our main results are then shown in Section~\ref{sec_main-results}.
Several other propositions are left to Appendix.

\section{Preliminaries}\label{sec_Preliminaries}

For a compact metric space $\Theta$, we naturally view\footnote{due to the fact that any reversible, deterministic classical computation may be represented by a permutation matrix on its register. See~Ref.~\cite[Section~20.2]{AB09} for a detailed explanation.} a classical oracle as a set of stochastic matrices 
\begin{equation*}
\NT = \{\Nt \text{~is a stochastic matrix} \mid \theta \in \Theta\}
\end{equation*}
 (with a fixed input and output dimensions independent of $\theta$) that are continuous with respect to a parameter $\theta \in \Theta$. 
A query oracle is a special case of this definition, since we can take $\Theta := \Bset^n$ and $\mathcal{N}_{x^n}: i \mapsto x_i$ for $x^n = (x_i)_{i \leq n}\in \Bset^n$.
Analogously in quantum scenario, a quantum oracle  is a set 
\begin{equation*}
\NT = \{\Nt \text{~is a quantum channel} \mid \theta \in \Theta\}
\end{equation*}
 of quantum channels (with a fixed input and output dimensions independent of $\theta$) that are continuous (as the diamond norm) with respect to a parameter $\theta \in \Theta$\footnote{Any norm on the space of quantum channels yields the same topology, since we are dealing with finite dimensional quantum systems.}.

To examine general oracle problems such as state/channel estimation processes, query complexity, discrimination problems in a unified manner, we define a target function to compute as a set of subsets in $\mathbb{R}^d$. 
Formally, a target function is defined as $\mathcal{F}_\Theta := \{F_\theta \subset \mathbb{R}^d \mid \theta \in \Theta\}$ for $d \geq 1$, and we say an algorithm computes $\mathcal{F}_\Theta$ when the output of the algorithm belongs to $F_\theta$ where $\theta \in \Theta$ denotes the parameter of the given oracle.
For example in the ordinary query scenario for computing a binary function $f:\{0, 1\}^n \to \Bset$, the target function is defined as $\mathcal{F}_\Theta = \{F_\theta = \{f(\theta)\} \subset \{0, 1\}\}$, in which each $F_\theta$ has exactly one element $f(\theta)$.

\par
For any classical or quantum algorithm $\pi$ for computing $\mathcal{F}(\Theta)$, let $|\pi|$ be the number of the worse-case oracle calls of the algorithm $\pi$ and $\E[\pi]$ be the expectation of the number of oracle calls over all possible randomness such as classical randomness and/or quantum measurements. We sometimes write $\E_\mu[\pi]$ to explicitly express the underlying distribution $\mu$ of inputs.

\subsection{Classical scenarios}

\paragraph{Distributional case}

A distributional oracle problem $P_D :=\PD$ is defined by a target function $\mathcal{F}_\Theta$,
a classical oracle $\NT$ and a distribution $\mu$ on $\Theta$.
$[P_D, \varepsilon] =\FTE$ denotes the set of oracle algorithms $\pi$ which try to output an element in $F_\theta$ with the error $ \Pr(\pi_\mathsf{out} \notin F_\theta ) \leq \varepsilon$ when the parameter $\theta \in \Theta$ is distributed according to $\mu$, where $\pi_\mathsf{out}$ denotes the output of the algorithm $\pi$.
Similarly, $[P_D, \varepsilon]^n = \FTEn$ denotes the set of oracle algorithms $\pi_n$ which compute $\mathcal{F}_\Theta^n = (\mathcal{F}_\Theta, \ldots,  \mathcal{F}_\Theta)$ with coordinate-wise error $\varepsilon$ when the parameter $\theta^n = (\theta_1, \ldots, \theta_n)$ is distributed according to $\mu^n$.

Define
\begin{equation*}
\ED(\PDE) := \inf_{\pi \in \PDE} \E[\pi], \quad
\ED(\PDEn) := \inf_{\pi_n \in \PDEn} \E[\pi_n], 
\end{equation*}
and
\begin{equation*}
D(\PDE) := \min_{\pi \in \PDE} |\pi|, \quad
D(\PDEn) := \min_{\pi_n \in \PDEn} |\pi_n|.
\end{equation*}

\paragraph{Randomized case}

A randomized oracle problem $P_R :=\PR$ is defined similarly to that of distributed oracle problems, except that a distribution $\mu$ on $\Theta$ does not appear in the randomized scenario.
$[P_R, \varepsilon] = [\PR, \varepsilon]$ denotes the set of oracle algorithms $\pi$ which compute $\mathcal{F}_\Theta$ with error $\Pr(\pi_\mathsf{out} \in F_\theta) \leq \varepsilon$ for any parameter $\theta \in \Theta$.
$[P_R, \varepsilon]^{n} = [\PR, \varepsilon]^n$ denotes the set of oracle algorithms $\pi_n$ which compute $\mathcal{F}^{n}_\Theta = (\mathcal{F}_\Theta, \ldots,  \mathcal{F}_\Theta)$ with coordinate-wise error $\varepsilon$ for any parameter $\theta \in \Theta$.
Analogously, 
\begin{equation*}
\ER(\PRE) := \inf_{\pi \in \PRE} \max_{\mu \in \PT}\E_{\mu}[\pi], \quad
\ER(\PREn) := \inf_{\pi_n \in \PREn}\max_{\mu^{\otimes n} \in \PT^n} \E_{\mu^{\otimes n}}[\pi_n]
\end{equation*}
where $\PT := \{\mu: \text{a probability distribution on $\Theta$}\}$.  We can also define
\begin{equation*}
\ER_D(\PRE) :=  \max_{\mu \in \PT}\inf_{\pi \in \PRE}\E_{\mu}[\pi].
\end{equation*}
Interestingly, by Proposition~\ref{Prop_minimax}, these values coincide: $\ER(\PRE)=\ER_D(\PRE)$.\footnote{The space $\PT$ is known to be compact w.r.t. the weak-$\ast$ topology.}
The randomized oracle complexity is defined in the ordinary way:
\begin{equation*}
R(\PRE) := \min_{\pi \in \PRE} |\pi|, \quad
R(\PREn) := \min_{\pi_n \in \PREn} |\pi_n|.
\end{equation*}

\subsection{Quantum scenarios}\label{subsec_quantum-scenarios}

\paragraph{Model of computation}
In this paper, we employ the model shown in Figure~\ref{Fig_quantum_comp} as a natural model of quantum computation with oracle access. 
This model seems quite similar to the ordinary one, except for the following two points.
One is that, at each round,  measurements are performed on some registers and decide whether another query access is required based on the outcomes.
This is a natural solution for dealing with the \emph{average-case} query complexity.
The other difference is that, before an execution of quantum processes, classical randomness $R$ is used to select which operators are performed in execution.
This is essentially for creating classical continuous random variables. 
In the classical scenario, time-unbounded circuits have the power of producing continuous random variables such as the uniform distribution on the interval $[0, 1]$. 
However, in quantum scenario, an infinite dimensional Hilbert space is required to produce such random variables, which causes several obstacles. 
To overcome such difficulties, classical randomness is attached in this model, and the whole quantum system remains finite-dimensional.
As general information, also note that the output $\pi_\mathsf{out}$ can be a quantum state or a classical output in this model.
\begin{figure}[hbtp]
\centering
	\includegraphics[keepaspectratio, scale=0.35]{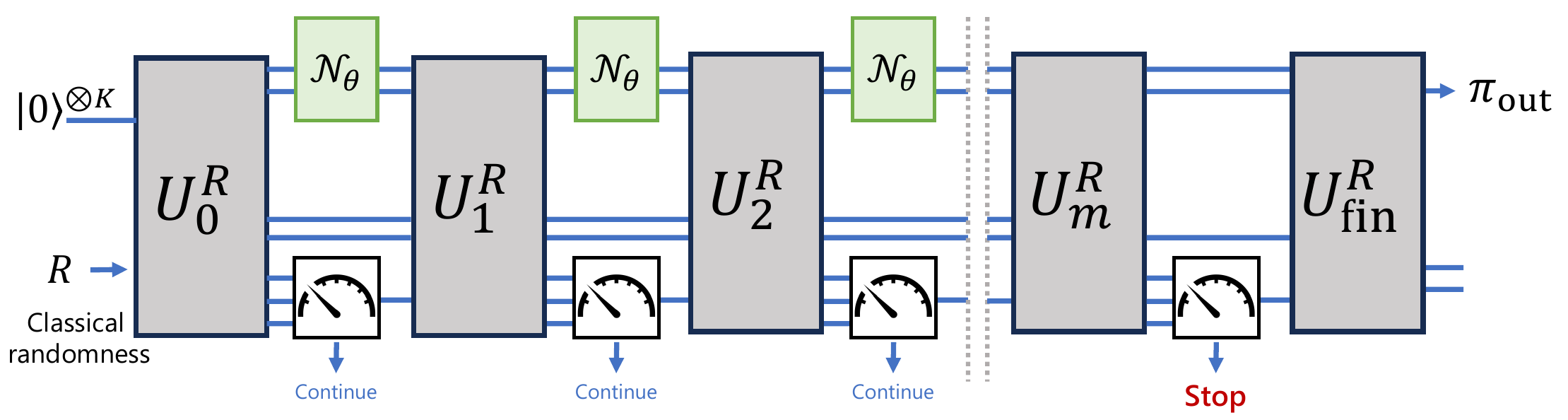}
	\caption{A general model of quantum computation with oracle access}
	\label{Fig_quantum_comp}
\end{figure}

In  case of the $n$ oracles $\{\mathcal{N}_{\theta_1}, \ldots, \mathcal{N}_{\theta_n}\}$, 
we focus on the model in which each selection of oracles is determined classically, as pictured in Figure~\ref{Fig_quantum_comp-n}.
In this model, each query access is selected by the classical randomness $R$ and the measurement outcome $M_i~(1 \leq i\leq m)$.
This model is weaker than the natural oracle model in which the selections of oracles are quantumly determined, i.e., 
the model where the oracle is defined by $\mathcal{N}_\mathsf{all} : |i \rangle \langle i| \otimes \rho \mapsto |i \rangle \langle i| \otimes \mathcal{N}_{\theta_i}(\rho)$.
\begin{figure}[hbtp]
\centering
	\includegraphics[keepaspectratio, scale=0.35]{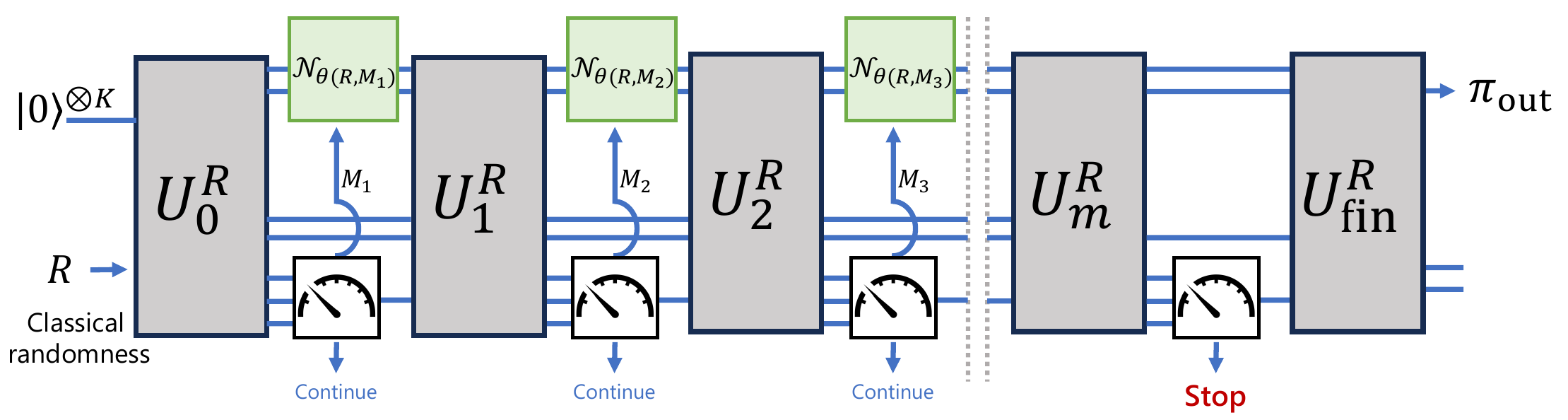}
	\caption{The model of quantum computation for $n$ oracles}
	\label{Fig_quantum_comp-n}
\end{figure}

\paragraph{Quantum distributed case}
A quantum distributed oracle problem expressed by $P_{QD} :=\PQD$ and 
the set of algorithms $[P_{QD}, \varepsilon] =\FQTE$ is defined similarly to that of classical distributed oracle problems. 
Note that in quantum scenarios, the output space $\mathcal{O}$ can be a quantum space.
$[P_{QD}, \varepsilon]^{n} = \FQTEn$ denotes the set of quantum oracle algorithms $\pi_n$ which compute $\mathcal{F}_\Theta^{n} = (\mathcal{F}_\Theta, \ldots, \mathcal{F}_\Theta)$ with coordinate-wise error $\varepsilon$ 
when the parameter $\theta^n = (\theta_1, \ldots, \theta_n)$ is distributed according to $\mu^n$.
Analogously, 
\begin{equation*}
\EQD([P_{QD}, \varepsilon]) := \inf_{\pi \in [P_{QD}, \varepsilon]} \E_{\mu}[\pi], \quad
\EQD([P_{QD}, \varepsilon]^{n}) := \inf_{\pi_n \in [P_{QD}, \varepsilon]^{n}} \E_{\mu^{\otimes n}}[\pi_n]
\end{equation*}
and
\begin{equation*}
QD([P_{QD}, \varepsilon]) := \min_{\pi \in [P_{QD}, \varepsilon]} |\pi|, \quad
QD([P_{QD}, \varepsilon]^{n}) := \min_{\pi_n \in [P_{QD}, \varepsilon]_\mathsf{Q}^{n}} |\pi_n|.
\end{equation*}

\paragraph{Quantum randomized case}
A quantum randomized oracle problem expressed by $P_{QR} :=\PR$ and
the set of algorithms $\PQRE = \FQRE$ is defined similarly to that of randomized oracle problems.
$\PQREn =\FQREn$ denotes the set of quantum oracle algorithms $\pi_n$ which compute $\mathcal{F}_\Theta^n = (\mathcal{F}_\Theta, \ldots, \mathcal{F}_\Theta)$ with coordinate-wise error $\varepsilon$ for any parameter $\theta \in \Theta$.
Analogously, 
\begin{equation*}
\EQR(\PQRE) := \inf_{\pi \in \PQRE} \max_{\mu \in \PT}\E_{\mu}[\pi], \quad
\EQR(\PQREn) := \inf_{\pi_n \in \PQREn}\max_{\mu^{\otimes n} \in \PT^n} \E_{\mu^{\otimes n}}[\pi_n]
\end{equation*}
and
\begin{equation*}
\EQR_D(\PQRE) :=  \max_{\mu \in \PT}\inf_{\pi \in \PQRE}\E_{\mu}[\pi].
\end{equation*}
Then by Proposition~\ref{Prop_minimax}, these values coincide: $\EQR(\PQRE)=\EQR_D(\PQRE)$.
The randomized oracle complexity is defined in the ordinary way:
\begin{equation*}
QR(\PQRE) := \min_{\pi \in \PQRE} |\pi|, \quad
QR(\PQREn) := \min_{\pi_n \in \PQREn} |\pi_n|.
\end{equation*}

\subsection{Examples in this framework}\label{subsec_examples}
Here we describe how our framework is applied to different complexity scenarios.

\paragraph{Classical/Quantum query complexity}
In this scenario, one aims to compute a (possibly promise function or relation) function  $f: \Bset^\ell \to \Bset$  efficiently.
In our framework, this scenario is represented by defining as follows:
\begin{itemize}
\item $\Theta := \Bset^\ell$ (or a subset of $\Bset^\ell$ in case for promise functions).
\item $F_\theta := \{f(\theta)\},~ (\theta \in \Theta = \Bset^\ell)$, the set that has one element $f(\theta)$. In case of relations, $F_\theta$ may have several different elements.
\item In classical case, $\mathcal{N}_x~(x \in \Bset^l)$ takes $i \in [l]$ as input and output $x_i$ w.p. exactly one. In quantum case, $\mathcal{N}_x: |i, a\rangle \mapsto |i\rangle|x_i \oplus a\rangle$ for $ a \in \Bset$.
\end{itemize}

\paragraph{Classical parameter estimation theory}
In this scenario, one aims to estimate a true parameter  $\theta \in \Theta~(\Theta \subset \mathbb{R}^d)$  efficiently,
when $\theta$ is unknown at the beginning but allowed to sample $x \in \mathcal{X}$ from a set $\mathcal{X}$ according to the distribution $x \sim p_\theta(x)$.
In our framework, this scenario is represented by defining as follows:
\begin{itemize}
\item $\Theta$ is the parameter space.
\item $F_\theta := \{\theta\},~ (\theta \in \Theta)$, the set that has one element $\theta$.
\item $\mathcal{N}_\theta$ takes nothing as input but output $x$ w.p. $p_\theta(x)$.
\end{itemize}

\paragraph{Quantum parameter estimation theory}
In this scenario, one aims to estimate a true parameter  $\theta \in \Theta~(\Theta \subset \mathbb{R}^d)$  efficiently,
when $\theta$ is unknown at the beginning but allowed to take a state $\rho_\theta$ arbitrarily many times.
In our framework, this scenario is represented by defining as follows:
\begin{itemize}
\item $\Theta$ is the parameter space.
\item $F_\theta := \{\theta\},~ (\theta \in \Theta)$, the set that has one element $\theta$.
\item $\mathcal{N}_\theta$ takes nothing as input but output $\rho_\theta$.
\end{itemize}

\paragraph{Classical PAC learning}
In this scenario, one beforehand knows an instance space $X$ and a set of possible concepts $\mathcal{C}$ that is a subset of the set of all concepts $\{c:X \to \Bset\}$.
Then the one aims to estimate some unknown concept $h \in \mathcal{C}$ with precision $1 - \delta$, by sampling only $h(x) \in \Bset$ where $x \in X$ obeys an unknown distribution $D \in \mathcal{D}$ on $X$.
In our framework, this scenario is represented by defining as follows:
\begin{itemize}
\item $\Theta := \mathcal{C} \times \mathcal{D}$.
\item $F_{(h, D)} := \{c \in \mathcal{C} \mid \Pr_D(h(x) \neq c(x)) \leq \delta\},~ ((h, D) \in \Theta )$.
\item $\mathcal{N}_\theta$ takes nothing as input but output $h(x)$ according the distribution $x \sim D$.
\end{itemize}

\section{Technical assumptions and facts}\label{sec_technical}
\begin{itemize}
\item For any algorithm $\pi$, $|\pi|$ is assumed to be finite.
\item For any small $\varepsilon > 0$, we assume $[P_C, \varepsilon]$  is not empty (and so is $[P_C, \varepsilon]^n$). This also implies that these sets can be empty when $\varepsilon = 0$. This condition becomes necessary when dealing with several instances such as parameter estimation processes. Note that we sometimes implicitly assume $[P_C, 0] \neq \emptyset$ when there is no confusion, such as in Lemma~\ref{Lem_algo_CD}.

\item The space $\mathcal{P}(\Theta)$ is formally defined as 
\begin{equation*}
\mathcal{P}(\Theta) := \{\text{a Borel probability measure}~ \mu \on \Theta \}.
\end{equation*}
\end{itemize}

The following facts come from functional analysis, specifically from the Banach-Alaoglu theorem~\cite[Theorem~3.15]{Rud91}.

\begin{Fact}
For any compact metric space $\Theta$, $\mathcal{P}(\Theta)$ is compact w.r.t. weak-$\ast$ topology.
\end{Fact}

\begin{Fact}
For any element $\theta_0$ in $\Theta$, the Dirac measure $\delta_{\theta_0}$ is an element of $\mathcal{P}(\Theta)$.
\end{Fact}

\begin{Fact}\label{Fact_conti_expe}
For any algorithm $\pi$,  its expectation $\E_\mu[\pi]$ and standard deviation $\sigma(\pi, \mu)$ are continuous w.r.t. $\mu \in \PT$.
\end{Fact}
\begin{proof}
First, observe that the probability of an algorithm $\pi$ finishing at the $i$th step is continuous due to the continuity of $\mathcal{N}_\theta$.
Let $\Pr_\theta(\text{$\pi$ finishes at the $i$th step})$ be the probability of an algorithm $\pi$ finishing at the $i$th step.
Then the expectation $\E_\theta[\pi]$, when the chosen parameter is $\theta$, is represented as
\begin{equation*}
\E_\theta[\pi] = \sum_{i \leq |\pi|} i \cdot \Pr_\theta(\text{$\pi$ finishes at the $i$th step})
\end{equation*}
which is a finite sum of continuous functions, and therefore $\E_\theta[\pi]$ is continuous w.r.t. $\theta$.
Since for any continuous function $f \in C(\Theta)$, $E_\mu[f]$ is continuous w.r.t. $\mu \in \mathcal{P}(\Theta)$, $\E_\mu[\pi]$ is continuous.
\par
For the standard deviation $\sigma(\pi, \mu)$, just use $\sigma^2(\pi, \mu) = \E_\mu[\pi^2] - \E_\mu^2[\pi]$.
\end{proof}

\section{Additivity}\label{sec_additivity}

\subsection{Classical distributional case}\label{subsec_add_CD}
\begin{Lem}\label{Lem_nleq_GO}
For any $P_D = \PD$ and $\varepsilon \in [0, 1]$,
$n\ED(\PDE)\leq \ED(\PDEn).$
\end{Lem}
\begin{proof}
Take $\pi_n \in \PDEn$ satisfying $\E(\pi_n) = \ED(\PDEn)$ (if there are no such algorithms, take a sequence converging to $\E(\PDEn)$.).
From $\pi_n$, we create $\tilde{\pi} \in \PDE$ with input $\theta \in \Theta$ as follows.
\begin{enumerate}
\item Pick $i \in [n]$ uniformly at random.
\item Privately pick $\tilde{\theta}^{n-1} \sim \mu^{n-1}$.
\item Run $\pi_n$ with input $(\tilde{\theta}, \ldots, \theta, \ldots, \tilde{\theta})$ in which $\theta$ is inserted at the $i$th position.
Note that every oracle access to $\tilde{\theta}^{n-1}$ is internally done without access to the actual oracle for $\theta$.
\end{enumerate}
This shows $\E[\tilde{\pi}]= \frac{1}{n}\E[\pi_n] =\frac{1}{n} \ED(\PDEn)$, which implies $n \ED(\PDE) \leq \ED(\PDEn)$.
This completes proof.
\end{proof}

\begin{Lem}\label{Lem_leqn_GO}
For any $P_D = \PD$ and $\varepsilon \in [0, 1]$,
$\E(\PDEn) \leq n\ED(\PDE) .$
\end{Lem}
\begin{proof}
Take $\pi \in \PDE$ satisfying $\E[\pi] =\ED(\PDE)$.
Let us create a new algorithm $\pi_n$ for $\mathcal{F}_\Theta^{n}$, by running $\pi$ repeatedly $n$ times.\footnote{Without written explicitly, any algorithm must be terminated with a finite number of oracle calls, infinitely many number of calls is not allowed in any algorithm.}
We see $\pi_n \in \PDEn$ and $\E[\pi_n] = n \E[\pi] = n\ED(\PDE)$.
Therefore, the definition of $\ED(\PDEn)$ implies $\ED(\PDEn)  \leq n\ED(\PDE)$.
This completes proof.
\end{proof}

Combining Lemma~\ref{Lem_nleq_GO} and Lemma~\ref{Lem_leqn_GO}, we get the additivity of $\ED(\PDEn)$:

\begin{Prop}\label{Prop_equal_GO}
$\ED(\PDEn) = n\ED(\PDE) .$
\end{Prop}

\subsection{Classical randomized case}\label{subsec_add_CR}

\begin{Lem}\label{Lem_nE_randomized_GO}
\begin{equation*}
n\ER(\PRE) \leq \ER(\PREn)
\end{equation*}
\end{Lem}
\begin{proof}
By Proposition~\ref{Prop_minimax}, we only need to show
$n\ER_D(\PRE) \leq \ER(\PREn)$.
For any $\delta >0$,  take $\pi_n \in \PREn$ such that
for any $\mu^{\otimes n} := \mu_1 \times \cdots \times \mu_n$, 
\begin{equation}\label{Lem_eq_ne_randomized_GO}
\E_{\mu^{\otimes n}}[\pi_n]
\leq \ER(\PREn) + \delta.
\end{equation}
Based on the algorithm $\pi_n$, we create $\tilde{\pi}_i \in \FE~(1 \leq i \leq n)$ as follows.
\begin{enumerate}
\item Privately pick $\tilde{\theta}^{n-1} \sim \mu_1 \times \cdots \times \mu_n$ except for $\mu_i$.
\item Run $\pi_n$ with input $(\tilde{\theta}, \ldots, \theta, \ldots, \tilde{\theta})$ in which the actual parameter $\theta$ is inserted at the $i$'th position.
Note that every oracle access to $\tilde{\theta}^{n-1}$ is internally done without access to the actual oracle for $\theta$.
\end{enumerate}
This construction of algorithms ensures that $\tilde{\pi}_i \in \PRE$ and furthermore, 
\begin{equation*}
\E_{\mu_1}[\tilde{\pi}_1]
+ \cdots +
\E_{\mu_n}[\tilde{\pi}_n]
=
\E_{\mu^{\otimes n}}[\tilde{\pi}_n].
\end{equation*}
Together with the inequality~\eqref{Lem_eq_ne_randomized_GO}, taking $\inf_{\tilde{\pi}_1, \ldots, \tilde{\pi}_n \in \PRE}$ and $\max_{\mu_1, \ldots, \mu_n}$ yields

\begin{equation*}
n\ER_D(\PRE) 
\leq \ER_D(\PREn) + \delta.
\end{equation*}
Since $\delta > 0$ is arbitrary, this completes proof.
\end{proof}

\begin{Lem}\label{Lem_En_randomized_GO}
\begin{equation*}
 \ER(\PREn) \leq
n\ER(\PRE) 
\end{equation*}
\end{Lem}

\begin{proof}
For any $\delta > 0$,
take $\pi \in \PRE$ satisfying that for any $\mu \in \PT$, $\E_\mu[\pi] \leq \ER(\PRE) + \delta$.
Let us create a new algorithm $\pi_n$ for $\mathcal{F}_\Theta^{ n}$, by running $\pi$ repeatedly for $n$ times.
We see $\pi_n \in \PREn$ and $\E_{\mu^{\otimes n}}[\pi_n] = \sum_{i \leq n}\E_{\mu_i}[\pi] \leq n(\ER(\PRE) +\delta)$ for any $\mu^{\otimes n} \in \PT^n$.
Therefore, the definition of $\ER(\PREn)$ implies $\ER(\PREn)  \leq n\ER(\PRE)$.
This completes proof.
\end{proof}

These two lemmas imply the following proposition.
\begin{Prop}\label{Prop_direct-sum_randomized_GO}
$\ER([P_R, \varepsilon]^{n}) =n \ER([P_R, \varepsilon])$.
\end{Prop}

\subsection{Quantum distributional case}\label{subsec_add_QD}

Since the proofs for the quantum distributed case are quite similar to the classical distributed case, we explain how to modify the proofs appropriately.
\begin{Lem}\label{Lem_nleq_QD}
For any $P_{QD} = \PQD$ and $\varepsilon \in [0, 1]$,
$n\EQD(\PQDE)\leq \EQD(\PQDEn).$
\end{Lem}
\begin{proof}
Modify Lemma~\ref{Lem_nleq_GO} straightforwardly. Note that picking elements $\tilde{\theta}^{n-1} \sim \mu^{n-1}$ is accomplished by using the classical randomness $R$ described in Section~\ref{subsec_quantum-scenarios}.
\end{proof}

\begin{Lem}\label{Lem_leqn_QD}
For any $P_{QD} = \PQD$ and $\varepsilon \in [0, 1]$,
$\EQD(\PQDEn) \leq n\EQD(\PQDE) .$
\end{Lem}
\begin{proof}
Modify Lemma~\ref{Lem_leqn_GO} straightforwardly.
\end{proof}

Combining Lemma~\ref{Lem_nleq_QD} and Lemma~\ref{Lem_leqn_QD}, we get the additivity of $\EQD(\PQDEn)$:

\begin{Prop}\label{Prop_equal_QD}
$\EQD(\PQDEn) = n\EQD(\PQDE) .$
\end{Prop}
\subsection{Quantum randomized case}\label{subsec_add_QR}

Modifying similarly to the quantum distributed case, we obtain the following statements.
\begin{Lem}\label{Lem_nleq_QR}
For any $P_{QR} = \PQR$ and $\varepsilon \in [0, 1]$,
$n\EQR(\PQRE)\leq \EQR(\PQREn).$
\end{Lem}

\begin{Lem}\label{Lem_leqn_QR}
For any $P_{QR} = \PQR$ and $\varepsilon \in [0, 1]$,
$\EQR(\PQREn) \leq n\EQR(\PQRE) .$
\end{Lem}

\begin{Prop}\label{Prop_equal_QR}
$\EQR(\PQREn) = n\EQR(\PQRE) .$
\end{Prop}

\section{Continuity}\label{sec_continuity}
Here we show the continuity with respect to the error parameter in $C([P_C, \varepsilon])$ for any positive $\varepsilon > 0$.
We also show the continuity at $\varepsilon = 0$ when the parameter space $\Theta$ is finite: $|\Theta| < \infty$.
\subsection{Classical distributional case}\label{subsec_cont_CD}

\begin{Lem}\label{Lem_conti_CD}
For any $\varepsilon > 0$,
\begin{equation*}
\lim_{\rho \to \varepsilon}\ED([P_D, \rho])
=\ED(\PDE).
\end{equation*}
\end{Lem}
\begin{proof}
Since both the limit $\rho \searrow \varepsilon$ and $\rho \nearrow \varepsilon$ are proved similarly, we only show the case $\rho \searrow \varepsilon$.
Take $\pi \in [P_D, \varepsilon]$ and create a new algorithm $\tilde{\pi} \in \PDE$ as running $\pi$ w.p. $ \varepsilon' := \varepsilon / (2\rho - \varepsilon)$ and another algorithm $\pi' \in [P_D, \varepsilon/2]$ w.p. $1 - \varepsilon'$.
The expectation of this algorithm satisfies
\begin{equation*}
\E[\tilde{\pi}]
= \varepsilon' \E[\pi] + (1- \varepsilon') \E[\pi'].
\end{equation*}
Let $\pi \in [P_D, \rho]$ be an optimal algorithm, i.e.,  $\E[\pi] = \ED([P_D, \rho])$.
Then the above equality implies
\begin{equation*}
\ED([P_D, \varepsilon])
\leq
\E[\tilde{\pi}]
=  \frac{\varepsilon}{2\rho - \varepsilon} \ED([P_D, \rho]) + \left(1-\frac{\varepsilon}{2\rho - \varepsilon}\right) \E[\pi'].
\end{equation*}
In addition, $\ED([P_D, \rho]) \leq \ED(\PDE)$ trivially holds.
Therefore, taking $\rho \searrow \varepsilon$ yields the desired statement.
\end{proof}

\begin{Lem}\label{Lem_conti_0_CD}
Suppose  $|\Theta| < \infty$.
Then 
\begin{equation*}
\ED([P_D, \alpha]) \leq \ED([P_D, \varepsilon]) + \sqrt{\varepsilon} ~\ED([P_D, \alpha/\sqrt{\varepsilon}])
\end{equation*}
holds for any $\alpha \in [0, 1)$ and any positive $\varepsilon$ satisfying
$ \sqrt{\varepsilon} < \mu_\textsf{min} := \min_{\theta \in \textrm{supp}\mu} \mu(\theta)$.

In particular the case\footnote{Note that $ [P_D, 0] \neq \emptyset$ is assumed here.} of $\alpha = 0$ shows that $\ED(\PDE)$ is also continuous at $\varepsilon =0$.
\end{Lem}
\begin{proof}
Define an algorithm for $P_D$ as follows.

\begin{center}
\resizebox{1.0\textwidth}{!}{
\begin{tabular}{|ll|} \hline
\multicolumn{2}{|c|}{\textbf{A new algorithm}}\\ \hline
1.&Run an optimal algorithm $\pi$ for $[P_D, \varepsilon]$, i.e., $\E[\pi] =\ED([P_D, \varepsilon])$.\\
&Let $M$ is the set of all possible patterns of the register after an execution of the algorithm $\pi$ and \\
&define $\mathcal{E}_m~(m \in M)$ as the event that $\pi$'s output is incorrect when the final register is $m \in M$.\\
2.&Output the original output $\pi_\mathsf{out}$ if $\Pr(\mathcal{E}_m) < \sqrt{\varepsilon}$.\\
3.&Otherwise run another algorithm $\pi'\in [P_D, \alpha/\sqrt{\varepsilon}]$ satisfying $\E[\pi'] = \ED([P_D, \alpha/\sqrt{\varepsilon}])$.\\
4.&Output $\pi_\mathsf{out}'$.\\ \hline
\end{tabular}
}
\end{center}

We now check the success probability and the cost of this algorithm.
\paragraph{Success probability:}
Assume $m$ satisfies $\Pr(\mathcal{E}_m) \leq \sqrt{\varepsilon}$.
Then $\Theta_m := \{\theta \in \Theta \mid \Pr(\theta \mid M = m) > 0\}$ satisfies $F(\Theta_m) \subset \{\textsf{Out}(m)\}$
where $\textsf{Out}(m)$ is the output represented in $m$. To see why, assume $\theta_0 \in \Theta_m$ satisfies $ \textsf{Out}(m)\notin F(\theta_0)$.
Then 
\begin{equation*}
\Pr(\mathcal{E}_m)
= \sum_{\theta \in \Theta} \mu(\theta) \Pr(\mathcal{E}_m| \Theta = \theta )
\geq \mu(\theta_0) \Pr(\mathcal{E}_m| \Theta = \theta_0)
\geq \mu_\textsf{min} > \sqrt{\varepsilon}
\end{equation*}
holds where the last inequality comes from the assumption that $ \textsf{Out}(m) \notin F(\theta_0)$ implying $\Pr(\mathcal{E}_m| \Theta = \theta_0) =1$.
This contradicts $\Pr(\mathcal{E}_m) \leq \sqrt{\varepsilon}$ and therefore $F(\Theta_m) \subset \{\textsf{Out}(m)\}$ holds.
Together with Markov inequality showing $\Pr(\Pr(\mathcal{E}_m) > \sqrt{\varepsilon}) \leq \sqrt{\varepsilon}$, this shows the error probability of the new algorithm is less than or equal to $\Pr(\Pr(\mathcal{E}_m \leq \sqrt{\varepsilon})) \cdot \alpha/\sqrt{\varepsilon} \leq \alpha$.

\paragraph{Expectation cost:}
Let us first check the probability of the algorithm terminating at the second step.
Again by Markov inequality $\Pr(\Pr(\mathcal{E}_m) > \sqrt{\varepsilon}) \leq \sqrt{\varepsilon}$ holds, which implies
\begin{align*}
\text{the probability}
= \sum_{\substack{m:\Pr(\mathcal{E}_m) \leq \sqrt{\varepsilon},\\ \theta \in \Theta_m}} \Pr(M = m) \Pr(\theta|M =m)
&= 1 - \Pr(\Pr(\mathcal{E}_m) > \sqrt{\varepsilon})\\
&\geq 1 - \sqrt{\varepsilon}.
\end{align*}
Therefore, the expectation cost $E$ satisfies
$E \leq \ED([P_D, \varepsilon]) + \sqrt{\varepsilon} ~ \ED([P_D, \alpha/\sqrt{\varepsilon}])$ which completes proof.
\end{proof}

\subsection{Classical randomized case}\label{subsec_cont_CR}
\begin{Lem}\label{Lem_conti_CR}
For any $\varepsilon \in (0, 1)$,
\begin{equation*}
\lim_{\rho \to \varepsilon} \ER([P_R, \rho])
= \ER(\FE).
\end{equation*}
\end{Lem}
\begin{proof}
Similar to Lemma~\ref{Lem_conti_CD}, we only show the limit $\rho \searrow \varepsilon$.
Take $\pi \in [P_R, \rho]$ and create a new algorithm $\tilde{\pi} \in \PRE$ as running $\pi$ w.p. $ \varepsilon' := \varepsilon / (2 \rho - \varepsilon)$ 
and another algorithm $\pi' \in [P_R, \varepsilon]$ w.p. $1 - \varepsilon'$. Note that $\varepsilon' \to 1$ as $\rho \to \varepsilon$.
The expectation of this algorithm satisfies
\begin{equation*}
\E_\mu[\tilde{\pi}]
= \varepsilon' \E_\mu[\pi] + (1- \varepsilon) \max_{\mu}\E_\mu[\pi']
\end{equation*}
for any distribution $\mu$.
Let $\pi \in [f, \rho]$ be an optimal algorithm, i.e.,  $\max_\mu \E_\mu[\pi] = \ER([f, \rho])$.
Then the above equality implies
\begin{equation*}
\ER(\PRE)
\leq
\max_\mu \E_\mu[\tilde{\pi}]
=  \varepsilon' \ER([P_R, \rho]) + (1-\varepsilon') \E_\mu[\tilde{\pi}].
\end{equation*}
Since trivially $\ER([P_R, \rho]) \leq \ER(\PRE)$,  taking $\rho \searrow \varepsilon$ yields the desired statement.
\end{proof}

\begin{Lem}\label{Lem_conti_0_CR}
Suppose  $|\Theta| < \infty$.
Then for any $\delta \in (0, 1)$ and any $\alpha \in [0, (\delta/4|\Theta|)^3)$,
\begin{equation*}
\left( 1 - \frac{3 \delta}{4 - 2 \delta}\right)\ER([P_R, 2\alpha|\Theta|/\delta]) \leq \ER([P_R, (\delta/4|\Theta|)^2])
\end{equation*}
holds.
In particular the case of $\alpha = 0$ shows that $\ER(\PRE)$ is also continuous at $\varepsilon =0$.
\end{Lem}
\begin{proof}
Using the same approach as in Lemma~\ref{Lem_conti_0_CD}, we immediately obtain that for any $\mu \in \mathcal{P}(\Theta)$, $\varepsilon < \mu^2_\mathsf{min}$, any algorithm $\pi \in \PRE$, there is an algorithm $\pi' \in [\PD, \alpha]$
such that
\begin{equation*}
\E_\mu[\pi'] \leq \E_\mu[\pi] + \sqrt{\varepsilon} ~\ER([P_R, \alpha/\sqrt{\varepsilon}]).
\end{equation*}
Define  $\bar{\mu} := (1 - \delta/2) \mu + \delta/2 \cdot U_\Theta$ for any $\delta \in (0, 1)$ and any $\mu \in \PT$,
 where $U_\Theta$ is the uniform distribution on $\Theta$.
Then the statement above implies that for any $\pi \in [P_R, \delta^2/16]$, there is an algorithm $\pi' \in [(\mathcal{F}_\Theta, \NT, \bar{\mu}), \alpha]$
such that
\begin{equation}\label{eq_Lem_conti_0_CR}
\E_{\bar{\mu}}[\pi'] \leq \E_{\bar{\mu}}[\pi] + \tilde{\delta}\ER([P_R, \alpha/\Tdelta]), \quad \tilde{\delta} :=\delta/4|\Theta|
\end{equation}
holds by $\bar{\mu}_\mathsf{min} \geq  2\tilde{\delta}$ and by substituting $\varepsilon := \tilde{\delta}^2$.
Now  the definition of $\bar{\mu}$ implies

\begin{align*}
\E_{\bar{\mu}}[\pi'] &= \left(1 - \frac{\delta}{2}\right)\E_{\mu}[\pi'] + \frac{\delta}{2}\E_{U_\Theta}[\pi']\\
&\geq \left(1 - \frac{\delta}{2}\right)\E_{\mu}[\pi'],\\
\E_{\bar{\mu}}[\pi] &= \left(1 - \frac{\delta}{2}\right)\E_{\mu}[\pi] + \frac{\delta}{2}\E_{U_\Theta}[\pi]\\
&\leq \left(1 - \frac{\delta}{2}\right)\E_{\mu}[\pi]+ \frac{\delta}{2}\max_{\mu \in \PT}\E_\mu[\pi].
\end{align*}
Together with the inequality~\eqref{eq_Lem_conti_0_CR}, these imply
\begin{equation}\label{eq_Lem_conti_0_CR_2}
 \left(1 - \frac{\delta}{2}\right)\E_{\mu}[\pi']
\leq \left(1 - \frac{\delta}{2}\right)\E_{\mu}[\pi]+ \frac{\delta}{2}\max_{\mu \in \PT}\E_\mu[\pi] +  \Tdelta \ER([P_R, \alpha/\Tdelta]).
\end{equation}
Since $\mu(\theta) \geq 2\Tdelta$ holds for any $\theta \in \Theta$, $\pi' \in [(\mathcal{F}_\Theta, \NT, \bar{\mu}), \alpha] \subset [P_R, \alpha/2\Tdelta]$ holds,
and therefore, taking $\inf_{\pi \in [P_R, \Tdelta^2]}$ and $\max_{\mu \in \PT}$ on the inequality~\eqref{eq_Lem_conti_0_CR_2} and Proposition~\ref{Prop_minimax} yields
\begin{align*}
 \left(1 - \frac{\delta}{2}\right)\ER([P_R, \alpha/2\Tdelta])
&\leq \left(1 - \frac{\delta}{2}\right)\ER([P_R, \Tdelta^2])+ \frac{\delta}{2}\ER([P_R, \Tdelta^2]) +  \Tdelta \ER([P_R, \alpha/\Tdelta])\\
&\leq \left(1 - \frac{\delta}{2}\right)\ER([P_R, \Tdelta^2])+ \frac{\delta}{2}\ER([P_R, \alpha/\Tdelta]) +  \Tdelta \ER([P_R, \alpha/\Tdelta])\\
&= \left(1 - \frac{\delta}{2}\right)\ER([P_R, \Tdelta^2])+ \frac{3\delta}{4}\ER([P_R, \alpha/\Tdelta]).
\end{align*}
Hence we obtain
\begin{equation*}
\ER([P_R, \alpha/2\Tdelta]) \leq \ER([P_R, \Tdelta^2])+  \frac{3\delta}{4 - 2 \delta}\ER([P_R, \alpha/\Tdelta]).
\end{equation*}
Considering $[P_R, \alpha/2\Tdelta] \subset [P_R, \alpha/\Tdelta]$, $\ER([P_R, \alpha/\Tdelta]) \leq \ER([P_R, \alpha/2\Tdelta])$ holds and therefore
\begin{equation*}
\ER([P_R, \alpha/2\Tdelta]) \leq \ER([P_R, \Tdelta^2])+  \frac{3\delta}{4 - 2 \delta}\ER([P_R, \alpha/2\Tdelta])
\end{equation*}
which completes proof.

\end{proof}

\subsection{Quantum distributional case}\label{subsec_cont_QD}

\begin{Lem}\label{Lem_conti_QD}
For any $\varepsilon > 0$,
\begin{equation*}
\lim_{\rho \to \varepsilon}\EQD([P_{QD}, \rho])
=\EQD(\PQDE).
\end{equation*}
\end{Lem}
\begin{proof}
Modify Lemma~\ref{Lem_conti_CD} straightforwardly.
\end{proof}

\begin{Lem}\label{Lem_conti_0_QD}
Suppose $[P_{QD}, 0] \neq \emptyset$ and $|\Theta| < \infty$.
We also assume the output of the problem $P_{QD}$ is classical. 
Then $\EQD(\PQDE)$ is also continuous at $\varepsilon =0$.
\end{Lem}
\begin{proof}
Modify Lemma~\ref{Lem_conti_0_CD} as follows.
Since the output needs to be classical, a measurement must be performed at the final step of computation to produce the output $\pi_\mathsf{out}$.
Without loss of generality, we assume the measurement is performed with the computational basis.
Define $M$ as the set of all possible patterns of extended measurement outcomes, 
which are obtained by performing the measurement with the computational basis to the entire quantum system, extending the measurement used originally in the algorithm $\pi$.
\par
The rest is shown in the same manner as in Lemma~\ref{Lem_conti_0_CD}.

\end{proof}
\subsection{Quantum randomized case}\label{subsec_cont_QR}
Similar to the quantum distributional case, we obtain the following lemmas.

\begin{Lem}\label{Lem_conti_QR}
For any $\varepsilon \in (0, 1)$,
\begin{equation*}
\lim_{\rho \to \varepsilon} \EQR([P_{QR}, \rho])
= \ER(\FE).
\end{equation*}
\end{Lem}

\begin{Lem}\label{Lem_conti_0_QR}
Suppose $[P_{QR}, 0] \neq \emptyset$ and $|\Theta| < \infty$.
We also assume the output of the problem $P_{QR}$ is classical. 
Then $\EQR(\PQRE)$ is also continuous at $\varepsilon =0$.
\end{Lem}

\section{Construction of optimal algorithms}\label{sec_opt}
\subsection{Classical distributional case}\label{subsec_opt_CD}

\begin{Lem}\label{Lem_algo_CD}
For any $n \in \mathbb{N}$, $\varepsilon >0$, $\alpha \in (0, \varepsilon)$, there is an algorithm $\pi \in \PDEn$ such that
\begin{equation*}
|\pi| \leq  n\ED([P_D, \varepsilon -\alpha]) + o(n).
\end{equation*}
This especially implies $D(\PDEn) \leq n\ED([P_D, \varepsilon -\alpha]) + o(n)$.

When $\varepsilon = 0$, for any $n \in \mathbb{N}$ and any $\alpha \in (0, 1)$, there is an algorithm $\pi \in [P_D^n, \alpha]$ such that
\begin{equation}\label{eq_Lem_algo_CD}
|\pi| \leq  n\ED([P_D, 0]) + o(n).
\end{equation}
This especially implies $D([P_D^n, \alpha]) \leq n\ED([P_D, 0]) + o(n)$.

\end{Lem}
\begin{proof}
For any $\alpha \in (0, \varepsilon)$, take the algorithm $\pi_\alpha^n \in [P_D, \varepsilon - \alpha]^{n}$ which is obtained by running an optimal algorithm $\pi \in [P_D, \varepsilon - \alpha]$ for $n$ times.
Let $\tilde{\pi}^n_{(\alpha, k)}~(\forall~k > 0)$ be a algorithm by terminating the algorithm $\pi_\alpha^n$ when the number of oracle calls reaches $\E[\pi_\alpha^n] + k \sigma_n = n\E[\pi_\alpha] + k \sigma_n$ where $\sigma_n$ is the standard deviation of the number of oracle calls in $\pi^n_\alpha$.
This definition implies $|\pi^n_{(\alpha, k)}| \leq n\E[\pi_\alpha] + k \sigma_n$, and, by Chebyshev's inequality $\Pr(|\pi_{(\alpha, k)}^n - n \E[\pi_\alpha]| \geq k \sigma) \leq k^{-2}$,
$\pi^n_{(\alpha, k)}$ computes $\mathcal{F}_\Theta^{\otimes n}$ with coordinate-wise error $\leq \varepsilon -\alpha + k^{-2}$.
This means $\pi^n_\alpha \in [P_D, \varepsilon -\alpha+ k^{-2}]^{n}$,
and therefore, 
\begin{equation}
|\pi_{(\alpha, k)}^n| \leq n \ED([P_D, \varepsilon -\alpha]) + k \sigma_n
\end{equation}
for any $k > 0$.
Substituting $k = 1/\sqrt{\alpha}$ yields
\begin{equation}
 |\pi_{(\alpha, \alpha^{-1/2})}^n|\leq  n\ED([P_D, \varepsilon -\alpha]) + \frac{\sigma_n}{\sqrt{\alpha}}.
\end{equation}
Since the standard deviation $\sigma_n$ of $n$-i.i.d. random variables scales as $\Theta(\sqrt{n})$, we obtain the desired argument.

\par
Similar proof works when $\varepsilon = 0$. First take an optimal algorithm $\pi^n \in [P^n_D, 0] = [P_D, 0]^n$ similarly and use Chebyshev's inequality.
Then set $k = 1/\sqrt{\alpha}$. The remaining algorithm satisfies the inequality~\eqref{eq_Lem_algo_CD}.
\end{proof}

\subsection{Classical randomized case}\label{subsec_opt_CR}
\begin{Lem}\label{Lem_algo_CR}
For any $n \in \mathbb{N}$, $\varepsilon >0$, $\alpha \in (0, \varepsilon)$, there is an algorithm $\pi \in \PREn$ such that
\begin{equation*}
|\pi| \leq  n\ER([P_R, \varepsilon -\alpha]) + o(n).
\end{equation*}
This especially implies $R(\PREn) \leq n\ER([P_R, \varepsilon -\alpha]) + o(n)$.

When $\varepsilon = 0$, for any $n \in \mathbb{N}$ and any $\alpha \in (0, 1)$, there is an algorithm $\pi \in [P_R^n, \alpha]$ such that
\begin{equation}\label{eq_Lem_algo_CR}
|\pi| \leq  n\ER([P_R, 0]) + o(n).
\end{equation}
This especially implies $R([P_R^n, \alpha]) \leq n\ER([P_R, 0]) + o(n)$.
\end{Lem}
\begin{proof}
For any $\alpha \in (0, \varepsilon)$ and any $\delta > 0$, 
take $\pi_\alpha \in [P_R, \varepsilon - \alpha]$ such that for any $\mu$, $\E_\mu[\pi_\alpha] < \ER([P_R, \varepsilon - \alpha]) + \delta$ holds.
The algorithm $\pi_\alpha^{\otimes n} \in [P_R, \varepsilon - \alpha]^{n}$ created by running $\pi_\alpha$ for $n$ times repeatedly satisfies
\begin{equation*}
\E_{\mu^{\otimes n}}[\pi_\alpha^{\otimes n}]
= \sum_{i \leq n} \E_{\mu_i} [\pi_\alpha].
\end{equation*}
for any $\mu^{\otimes n} = \mu_1 \times \cdots \times \mu_n \in \PT$. By Chebyshev's inequality, for any $k > 0$,
\begin{equation}\label{Thm_eq_Che_random_GO}
\Pr_{\mu^{\otimes n}}(|\pi^n_\alpha| - \E_{\mu^{\otimes n}}[\pi^n_\alpha] \geq k \sigma_n(\pi^n_\alpha, \mu^{\otimes n}))
\leq \frac{1}{k^2}.
\end{equation}
Considering that the standard deviation $\sigma_n(\pi^n_\alpha, \mu^{\otimes n})$ is calculated as
\begin{equation*}
\sigma_n(\pi^n_\alpha, \mu^{\otimes n})
= \sqrt{\sum_{i \leq n} \sigma^2(\pi_\alpha, \mu_i)},
\end{equation*}
Chebyshev's inequality~\eqref{Thm_eq_Che_random_GO} further implies
\begin{equation*}
\Pr_{\mu^{\otimes n}}(|\pi^n_\alpha|  \geq n \max_\mu \E_{\mu}[\pi_\alpha] + k \max_\mu \sqrt{n}\sigma(\pi_\alpha, \mu))
\leq \frac{1}{k^2}.
\end{equation*}
(Note that both $\E_\mu[\pi_\alpha]$ and $\sigma(\pi_\alpha, \mu)$ are continuous on $\mu \in \PT$ from Fact~\ref{Fact_conti_expe}.)
Next we define $\pi^n_{(\alpha, k)}$ as the algorithm $\pi^n_\alpha$ with the additional condition that it must be terminated when the number of query calls reaches
$n \max_\mu \E_{\mu}[\pi_\alpha] + k \max_\mu \sqrt{n}\sigma(\pi_\alpha, \mu)$.
For any $\theta^0_1, \ldots, \theta^0_n \in \Theta$, take $\mu^{\otimes n}$ as $\mu_i(\theta_i^0) =1~(1 \leq i \leq n)$, and Chebyshev's inequality shows
the error probability of the algorithm $\pi^n_{(\alpha, k)}$ on the parameter $(\theta^0_1, \ldots, \theta^0_n)$ is at most $\varepsilon - \alpha + 1/k^2$.
That is, $\pi^n_{(\alpha, k)}$ is an element of $[P_R, \varepsilon - \alpha +1/k^2]^{n}$. Therefore, by substituting $k = 1/\sqrt{\alpha}$,
we get $\pi^n_{(\alpha, 1/\sqrt{\alpha})} \in [P_R, \varepsilon]^n$ and 
\begin{align*}
|\pi^n_{(\alpha, 1/\sqrt{\alpha})}| &\leq 
n \max_\mu \E_{\mu}[\pi_\alpha] + \frac{1}{\sqrt{\alpha}} \max_\mu \sqrt{n}\sigma(\pi_\alpha, \mu)\\
&\leq n (\ER([P_R, \varepsilon -\alpha]) + \delta) + \frac{1}{\sqrt{\alpha}} \max_\mu \sqrt{n}\sigma(\pi_\alpha, \mu).
\end{align*}
This shows the desired statement. In the case of $\varepsilon = 0$, apply a similar argument given in Lemma~\ref{Lem_algo_CD}.
\end{proof}

\subsection{Quantum cases}\label{subsec_opt_QD}
Modifying Lemma~\ref{Lem_algo_QD} and Lemma~\ref{Lem_algo_QR} straightforwardly, we obtain the following lemmas for quantum scenarios.

\paragraph{Quantum distribution case}
\begin{Lem}\label{Lem_algo_QD}
For any $n \in \mathbb{N}$, $\varepsilon >0$, $\alpha \in (0, \varepsilon)$, there is an algorithm $\pi \in \PQDEn$ such that
\begin{equation*}
|\pi| \leq  n\EQD([P_{QD}, \varepsilon -\alpha]) + o(n).
\end{equation*}
This especially implies $QD(\PDEn) \leq n\EQD([P_D, \varepsilon -\alpha]) + o(n)$.

When $\varepsilon = 0$, for any $n \in \mathbb{N}$ and any $\alpha \in (0, 1)$, there is an algorithm $\pi \in [P_{QD}^n, \alpha]$ such that
\begin{equation}
|\pi| \leq  n\EQD([P_{QD}, 0]) + o(n).
\end{equation}
This especially implies $QD([P_{QD}^n, \alpha]) \leq n\EQD([P_{QD}, 0]) + o(n)$.
\end{Lem}

\paragraph{Quantum randomized case}
\begin{Lem}\label{Lem_algo_QR}
For any $n \in \mathbb{N}$, $\varepsilon >0$, $\alpha \in (0, \varepsilon)$, there is an algorithm $\pi \in \PQREn$ such that
\begin{equation*}
|\pi| \leq  n\EQR([P_{QR}, \varepsilon -\alpha]) + o(n).
\end{equation*}
This especially implies $QR(\PQREn) \leq n\EQR([P_{QR}, \varepsilon -\alpha]) + o(n)$.
 
When $\varepsilon = 0$, for any $n \in \mathbb{N}$ and any $\alpha \in (0, 1)$, there is an algorithm $\pi \in [P_{QR}^n, \alpha]$ such that
\begin{equation}
|\pi| \leq  n\EQR([P_{QR}, 0]) + o(n).
\end{equation}
This especially implies $QR([P_{QR}^n, \alpha]) \leq n\EQR([P_{QR}, 0]) + o(n)$.
\end{Lem}

\section{Main results}\label{sec_main-results}
Here our main results, Theorem~\ref{Thm_unified_amortized}, \ref{Thm_unified_amortized_Theta} and \ref{Thm_unified_amortized_Theta}, are proved applying the statements shown in the previous sections. Theorem~\ref{Thm_unified_amortized} and \ref{Thm_unified_amortized_Theta} deal with the direct sum theorems \emph{with the limit}, while Theorem~\ref{Thm_direct_sum_expec_0} deals with the theorems \emph{without limit}. Additionally, several propositions are proved in this section while completing the proofs for main results. These propositions may be of independent interest.

In the below we first focus on the proof of Theorem~\ref{Thm_unified_amortized}.
\begin{Thm1}
For any complexity scenario $C \in \{D, R, QD, QR\}$, any $\varepsilon > 0$, and any problem $P_C$,
\begin{equation*}
\lim_{n \to \infty} \frac{C([P_C, \varepsilon]^n)}{n}
=
\EC([P_C, \varepsilon]).
\end{equation*}
\end{Thm1}
\begin{proof}
As proved in Section~\ref{sec_additivity}, 
\begin{equation*}
\EC([P_C, \varepsilon])
\leq
\lim_{n \to \infty} \frac{C([P_C, \varepsilon]^n)}{n}
\end{equation*}
holds. Also as in Section~\ref{sec_opt}, 
\begin{equation*}
\lim_{n \to \infty} \frac{C([P_C, \varepsilon]^n)}{n}
\leq
\EC([P_C, \varepsilon - \alpha])
\end{equation*}
holds for any $\alpha > 0$. Together with the continuity in Section~\ref{sec_continuity}, 
\begin{equation*}
\EC([P_C, \varepsilon])
\leq
\lim_{n \to \infty} \frac{C([P_C, \varepsilon]^n)}{n}
\leq
\lim_{\alpha \to 0}\EC([P_C, \varepsilon - \alpha])
= \EC([P_C, \varepsilon]).
\end{equation*}
This completes the proof.
\end{proof}

The following proposition gives the additivity that works for any complexity scenario.
\begin{Prop}\label{Prop_unified_additivity}
For any complexity scenario $C \in \{D, R, QD, QR\}$, for any $\varepsilon \geq 0$, for any problem $P_C$,
\begin{equation*}
\EC([P_C, \varepsilon]^n) = n\EC([P_C, \varepsilon]).
\end{equation*}
\end{Prop}
\begin{proof}
This is immediate from results in Section~\ref{sec_additivity}.
\end{proof}

Proposition~\ref{Prop_direct_sum_expec} deals with direct sum theorems when the overall error is small.
\begin{Prop}\label{Prop_direct_sum_expec}
Suppose $|\Theta| < \infty$.
\begin{enumerate}[(i)]
\item In any distributional problem $P_C$~($C \in \{D, QD\}$) 
with non-trivial distribution\footnote{A distribution $\mu$ is non-trivial if and only if $\max_{\theta \in \Theta} \mu(\theta) < 1$} $\mu$,  
\begin{equation*}
\EC([P_C^n, \varepsilon]) = \Theta(n \cdot \EC([P_C, \varepsilon/n]))
\end{equation*}
for any $n \in \mathbb{N}$ and any positive
$ \varepsilon < \min\{99/100, \mu^2_\textsf{min}\}$ where the value $\mu_\textsf{min}$ is defined as $\mu_\textsf{min} := \min_{\theta \in \textrm{supp}\mu} \mu(\theta)$.

\item In any randomized problem $P_C$~($C \in \{R, QR\}$), 
\begin{equation*}
\EC([P_C^n, \varepsilon]) = \Theta(n \cdot \EC([P_C, \varepsilon/n]))
\end{equation*}
for any $n \in \mathbb{N}$ and any $\varepsilon \in (0, 1/128|\Theta|^2)$.
\end{enumerate}
\end{Prop}

\begin{proof}
We first show
$\EC([P_C^n, \varepsilon]) = \Omega(n \cdot \EC([P_C, \varepsilon/n]))$.
For any complexity scenario, by Proposition~\ref{Prop_unified_additivity},
\begin{equation}\label{Thm_direct_expec_eq_1}
n \EC([P_C, \varepsilon]) \leq \EC([P_C, \varepsilon]^n) \leq \EC([P_C^n, \varepsilon])
\end{equation}
holds. For distributional problems, we further obtain by substituting $\alpha = \varepsilon^{3/2}/n$ in Lemma~\ref{Lem_conti_0_CD} 
\begin{equation*}
(1 - \sqrt{\varepsilon})\EC([P_C, \varepsilon/n])
\leq \EC([P_C, \varepsilon])
\end{equation*}
for any $n \in \mathbb{N}$ and any positive
$ \varepsilon < \min\{99/100, \mu^2_\textsf{min}\}$.  This implies
$\EC([P_C, \varepsilon]) = \Omega(n\cdot\EC([P_C, \varepsilon/n]))$ together with
the inequality~\eqref{Thm_direct_expec_eq_1}.
On the other hand, in case of randomized problems, we obtain
by Lemma~\ref{Lem_conti_0_CR}
\begin{equation*}
\left(1 - \frac{3\delta}{4-2\delta}\right)\EC([P_C, 2\alpha/\delta])
\leq \EC([P_C, \delta^2/16])
\end{equation*}
which is simplified to,
by substituting $\delta = 1/2$ and $\alpha = \varepsilon/4n|\Theta|$,
\begin{equation*}
\frac{1}{2}\cdot\EC([P_C, \varepsilon/n])
\leq \EC([P_C, 1/64]) \leq 
\EC([P_C, \varepsilon])
\end{equation*}
for any $\varepsilon < 1/128|\Theta|^2$\footnote{This condition comes from $4 \alpha = \varepsilon/4n|\Theta| < 4 (\delta/4|\Theta|)^3$.}. Together with the inequality~\eqref{Thm_direct_expec_eq_1}, 
we have 
$\EC([P^n_C, \varepsilon]) = \Omega(n \cdot \EC([P_C, \varepsilon/n]))$.
\par
To show the other inequality:
$\EC([P_C^n, \varepsilon]) = O(n \cdot \EC([P_C, \varepsilon/n]))$, observe that $ [P_C, \varepsilon/n]^n$ is contained in$ [P_C^n, \varepsilon]$,  and therefore
\begin{equation*}
\EC([P_C^n, \varepsilon]) 
\leq
\EC([P_C, \varepsilon/n]^n) 
= n \EC([P_C, \varepsilon/n]).
\end{equation*}
This completes proof.
\end{proof}

In Theorem~\ref{Thm_direct_sum_expec_0}, we show direct sum theorems when the overall error is small in terms of the expected query/oracle complexity.
\begin{Thm3}
Suppose $|\Theta| < \infty$ and $[P_C, 0] \neq \emptyset$.
\begin{enumerate}[(i)]
\item If $P_C$ is a  distributional problem~(i.e., $C \in \{D, QD\}$)
with non-trivial distribution $\mu$,
\begin{equation*}
\EC([P_C^n, \varepsilon]) = \Theta(n \cdot \EC([P_C, 0]))
\end{equation*}
for any $n \in \mathbb{N}$ and any positive
$ \varepsilon < \min\{99/100, \mu^2_\textsf{min}\}$ where  the value $\mu_\textsf{min}$ is defined as $\mu_\textsf{min} := \min_{\theta \in \textrm{supp}\mu} \mu(\theta)$.

\item If $P_C$ is a randomized problem~(i.e., $C \in \{R, QR\}$), 
\begin{equation*}
\EC([P_C^n, \varepsilon]) = \Theta(n \cdot \EC([P_C, 0]))
\end{equation*}
for any $n \in \mathbb{N}$ and any $\varepsilon \in (0, 1/128|\Theta|^2)$.
\end{enumerate}
\end{Thm3}

\begin{proof}
Take $\alpha = 0 $ in Lemma~\ref{Lem_conti_0_CD} or Lemma~\ref{Lem_conti_0_CR} and the rest is shown in the same manner as in Proposition~\ref{Prop_direct_sum_expec}.
\end{proof}

In contrast to  Theorem~\ref{Thm_direct_sum_expec_0}, we below show direct sum theorems when the overall error is small in terms of the worst-case query/oracle complexity.
\begin{Thm2}
Suppose $|\Theta| < \infty$ and $[P_C, 0] \neq \emptyset$.
Then for any complexity scenario $C \in \{D, R, QD, QR\}$, for any problem $P_C$,
\begin{equation*}
\lim_{n \to \infty} \frac{C([P^n_C, \varepsilon])}{n}
=
\Theta\left(\EC([P_C, 0])\right).
\end{equation*}
for any positive $\varepsilon < 1/128|\Theta|^2$ if $C \in \{R, QR\}$ and any positive $\varepsilon < \mu^2_\mathsf{min}$ if $C \in \{D, QD\}$.
\end{Thm2}
\begin{proof}
To show $\lim_{n \to \infty} \frac{C([P^n_C, \varepsilon])}{n} = \Omega(\EC([P_C, 0]))$,  use $\EC(P_C^n, \varepsilon) \leq C([P_C^n, \varepsilon])$ and Theorem~\ref{Thm_direct_sum_expec_0}.
To show $\lim_{n \to \infty} \frac{C([P^n_C, \varepsilon])}{n} = O(\EC([P_C, 0]))$, use the results: $C([P_C^n, \alpha]) \leq n\EC([P_C, 0]) + o(n)$ proved in Section~\ref{sec_opt}.
\end{proof}

\section*{Acknowledgement}
The author would like to thank Shun Watanabe for valuable and insightful discussions during three online meetings over 10 hours in total, without whom this work could not have been completed.
The author would also like to thank his supervisor Fran\c{c}ois Le Gall for his kindness and valuable comments regarding an earlier draft of this paper.
The author would also like to thank his friend Ziyu Liu for his friendship and pointing out the reference~\cite{Rud91}.
The author was partially supported by MEXT Q-LEAP grants No. JPMXS0120319794.

\bibliography{/Users/surugadaiki/Dropbox/Citations/mathematics_D, /Users/surugadaiki/Dropbox/Citations/computational_D, /Users/surugadaiki/Dropbox/Citations/query_D, /Users/surugadaiki/Dropbox/Citations/quant_info_D, /Users/surugadaiki/Dropbox/Citations/comm_comp_D, /Users/surugadaiki//Dropbox/Citations/books_D}

\begin{thebibliography}{10}

\bibitem{JKS10}
Rahul Jain, Hartmut Klauck, and Miklos Santha.
\newblock Optimal direct sum results for deterministic and randomized decision
  tree complexity.
\newblock {\em Information Processing Letters}, 110(20):893--897, 2010.

\bibitem{ACLT10}
Andris Ambainis, Andrew~M Childs, Fran{\c{c}}ois Le~Gall, and Seiichiro Tani.
\newblock The quantum query complexity of certification.
\newblock {\em Quantum Information \& Computation}, 10(3):181--189, 2010.

\bibitem{Dru11}
Andrew Drucker.
\newblock Improved direct product theorems for randomized query complexity.
\newblock In {\em 26th Computational Complexity Conference}, pages 1--11, 2011.

\bibitem{Mon13}
Ashley Montanaro.
\newblock A composition theorem for decision tree complexity.
\newblock {\em arXiv preprint arXiv:1302.4207}, 2013.

\bibitem{MS15}
Sagnik Mukhopadhyay and Swagato Sanyal.
\newblock Towards better separation between deterministic and randomized query
  complexity.
\newblock In {\em 35th IARCS Annual Conference on Foundations of Software
  Technology and Theoretical Computer Science}, page 206, 2015.

\bibitem{BK18}
Shalev Ben-David and Robin Kothari.
\newblock Randomized query complexity of sabotaged and composed functions.
\newblock {\em Theory of Computing}, 14(1):1--27, 2018.

\bibitem{BB19}
Eric Blais and Joshua Brody.
\newblock Optimal separation and strong direct sum for randomized query
  complexity.
\newblock In {\em 34th Computational Complexity Conference}, pages 1--17, 2019.

\bibitem{BB20}
Shalev Ben-David and Eric Blais.
\newblock A tight composition theorem for the randomized query complexity of
  partial functions.
\newblock In {\em 61st annual Symposium on Foundations of Computer Science},
  pages 240--246, 2020.

\bibitem{GM21}
Mika G{\"o}{\"o}s and Gilbert Maystre.
\newblock A majority lemma for randomised query complexity.
\newblock In {\em 36th Computational Complexity Conference}, 2021.

\bibitem{BTLS23}
Joshua Brody, Jae~Tak Kim, Peem Lerdputtipongporn, and Hariharan Srinivasulu.
\newblock A strong {XOR} lemma for randomized query complexity.
\newblock {\em Theory of Computing}, 19(1):1--14, 2023.

\bibitem{BKST24}
Guy Blanc, Caleb Koch, Carmen Strassle, and Li-Yang Tan.
\newblock {A Strong Direct Sum Theorem for Distributional Query Complexity}.
\newblock In {\em 39th Computational Complexity Conference}, volume 300, pages
  16:1--16:30, 2024.

\bibitem{FKNN95}
Tom{\'a}s Feder, Eyal Kushilevitz, Moni Naor, and Noam Nisan.
\newblock Amortized communication complexity.
\newblock {\em SIAM Journal on computing}, 24(4):736--750, 1995.

\bibitem{CSWY01}
Amit Chakrabarti, Yaoyun Shi, Anthony Wirth, and Andrew Yao.
\newblock Informational complexity and the direct sum problem for simultaneous
  message complexity.
\newblock In {\em 42nd annual Symposium on Foundations of Computer Science},
  pages 270--278, 2001.

\bibitem{JRS03d}
Rahul Jain, Jaikumar Radhakrishnan, and Pranab Sen.
\newblock A direct sum theorem in communication complexity via message
  compression.
\newblock In {\em 30th International Colloquium on Automata, Languages and
  Programming}, pages 300--315, 2003.

\bibitem{BJKS04}
Ziv Bar-Yossef, T.S. Jayram, Ravi Kumar, and D.~Sivakumar.
\newblock An information statistics approach to data stream and communication
  complexity.
\newblock {\em Journal of Computer and System Sciences}, 68(4):702--732, 2004.

\bibitem{BBCR09}
Boaz Barak, Mark Braverman, Xi~Chen, and Anup Rao.
\newblock Direct sums in randomized communication complexity.
\newblock {\em Electronic Colloquium on Computational Complexity}, 44, 2009.

\bibitem{JK09}
Rahul Jain and Hartmut Klauck.
\newblock New results in the simultaneous message passing model via information
  theoretic techniques.
\newblock In {\em 24th Computational Complexity Conference}, pages 369--378,
  2009.

\bibitem{BR11}
Mark Braverman and Anup Rao.
\newblock Information equals amortized communication.
\newblock In {\em 52nd annual Symposium on Foundations of Computer Science},
  pages 748--757, 2011.

\bibitem{JPY12}
Rahul Jain, Attila Pereszl{\'e}nyi, and Penghui Yao.
\newblock A direct product theorem for the two-party bounded-round public-coin
  communication complexity.
\newblock In {\em 53rd annual Symposium on Foundations of Computer Science},
  pages 167--176, 2012.

\bibitem{MWY13}
Marco Molinaro, David~P Woodruff, and Grigory Yaroslavtsev.
\newblock Beating the direct sum theorem in communication complexity with
  implications for sketching.
\newblock In {\em 24th annual ACM-SIAM symposium on Discrete algorithms}, pages
  1738--1756, 2013.

\bibitem{KMSY14}
Gillat Kol, Shay Moran, Amir Shpilka, and Amir Yehudayoff.
\newblock Direct sum fails for zero error average communication.
\newblock In {\em 5th conference on Innovations in Theoretical Computer
  Science}, pages 517--522, 2014.

\bibitem{Bra17}
Mark Braverman.
\newblock Interactive information complexity.
\newblock {\em SIAM Review}, 59(4):803--846, 2017.

\bibitem{Jain20}
Rahul Jain.
\newblock A near-optimal direct-sum theorem for communication complexity.
\newblock {\em arXiv preprint arXiv:2008.07188}, 2020.

\bibitem{JK22}
Rahul Jain and Srijita Kundu.
\newblock A direct product theorem for quantum communication complexity with
  applications to device-independent {QKD}.
\newblock In {\em 62nd annual Symposium on Foundations of Computer Science},
  pages 1285--1295, 2022.

\bibitem{Wu22}
Hao Wu.
\newblock Direct sum theorems from fortification.
\newblock {\em arXiv preprint arXiv:2208.07730}, 2022.

\bibitem{IW97}
Russell Impagliazzo and Avi Wigderson.
\newblock {P= BPP} if {E} requires exponential circuits: Derandomizing the xor
  lemma.
\newblock In {\em 29th annual Symposium on Theory of Computing}, pages
  220--229, 1997.

\bibitem{IJKW08}
Russell Impagliazzo, Ragesh Jaiswal, Valentine Kabanets, and Avi Wigderson.
\newblock Uniform direct product theorems: simplified, optimized, and
  derandomized.
\newblock In {\em 40th annual Symposium on Theory of Computing}, pages
  579--588, 2008.

\bibitem{GNW11}
Oded Goldreich, Noam Nisan, and Avi Wigderson.
\newblock On {Yao}'s {XOR}-lemma.
\newblock {\em Studies in Complexity and Cryptography}, pages 273--301, 2011.

\bibitem{Pan12}
Denis Pankratov.
\newblock Direct sum questions in classical communication complexity.
\newblock {\em Master's thesis, University of Chicago}, 2012.

\bibitem{Sha01}
Ronen Shaltiel.
\newblock Towards proving strong direct product theorems.
\newblock In {\em 16th Computational Complexity Conference}, pages 107--117,
  2001.

\bibitem{Gro96}
Lov~K Grover.
\newblock A fast quantum mechanical algorithm for database search.
\newblock In {\em 28th annual Symposium on Theory of Computing}, pages
  212--219, 1996.

\bibitem{LMR+11}
Troy Lee, Rajat Mittal, Ben~W Reichardt, Robert {\v{S}}palek, and Mario
  Szegedy.
\newblock Quantum query complexity of state conversion.
\newblock In {\em 52nd annual Symposium on Foundations of Computer Science},
  pages 344--353, 2011.

\bibitem{KN96}
Eyal Kushilevitz and Noam Nisan.
\newblock {\em Communication Complexity}.
\newblock Cambridge University Press, 1996.

\bibitem{RY20}
Anup Rao and Amir Yehudayoff.
\newblock {\em Communication Complexity: and Applications}.
\newblock Cambridge University Press, 2020.

\bibitem{Bra15}
Mark Braverman.
\newblock Interactive information complexity.
\newblock {\em SIAM Journal on Computing}, 44(6):1698--1739, 2015.

\bibitem{Tou15}
Dave Touchette.
\newblock Quantum information complexity.
\newblock In {\em 47th Annual Symposium on Theory of Computing}, pages
  317--326, 2015.

\bibitem{BEO+13}
Mark Braverman, Faith Ellen, Rotem Oshman, Toniann Pitassi, and Vinod
  Vaikuntanathan.
\newblock A tight bound for set disjointness in the message-passing model.
\newblock In {\em 54th annual Symposium on Foundations of Computer Science},
  pages 668--677, 2013.

\bibitem{TVVW17}
Himanshu Tyagi, Shaileshh~Bojja Venkatakrishnan, Pramod Viswanath, and Shun
  Watanabe.
\newblock Information complexity density and simulation of protocols.
\newblock {\em IEEE Transactions on Information Theory}, 11(63):6979--7002,
  2017.

\bibitem{BGK+18}
Mark Braverman, Ankit Garg, Young~Kun Ko, Jieming Mao, and Dave Touchette.
\newblock Near-optimal bounds on the bounded-round quantum communication
  complexity of disjointness.
\newblock {\em SIAM Journal on Computing}, 47(6):2277--2314, 2018.

\bibitem{ABB+16}
Andris Ambainis, Kaspars Balodis, Aleksandrs Belovs, Troy Lee, Miklos Santha,
  and Juris Smotrovs.
\newblock Separations in query complexity based on pointer functions.
\newblock In {\em 48th annual Symposium on Theory of Computing}, pages
  800--813, 2016.

\bibitem{Neu28}
J.~von Neumann.
\newblock Zur theorie der gesellschaftsspiele.
\newblock {\em Mathematische Annalen}, 100(1):295--320, 1928.

\bibitem{BGPW13}
Mark Braverman, Ankit Garg, Denis Pankratov, and Omri Weinstein.
\newblock From information to exact communication.
\newblock In {\em 45th annual ACM Symposium on Theory of computing}, pages
  151--160, 2013.

\bibitem{BB23}
Shalev Ben-David and Eric Blais.
\newblock A new minimax theorem for randomized algorithms.
\newblock {\em Journal of the ACM}, 70(6):1--58, 2023.

\bibitem{AB09}
Sanjeev Arora and Boaz Barak.
\newblock {\em Computational Complexity: A Modern Approach}.
\newblock Cambridge University Press, 2009.

\bibitem{Rud91}
Walter Rudin.
\newblock {\em Functional Analysis}.
\newblock International series in pure and applied mathematics. McGraw-Hill,
  1991.

\end{thebibliography}
\bibliographystyle{unsrt}
\appendix
\section{Supplemental materials}

\begin{Prop}\label{Prop_appendix_cont}
For a function $f: \Bset^n \to \Bset$, let $\varepsilon < 2^{-n}$. Then $R([f, \varepsilon]) = R([f, 0])$.
\end{Prop}

\begin{proof}
Let $\pi \in [f, \varepsilon]$ be an optimal algorithm: $|\pi| = R([f, \varepsilon])$ and denote the set of randomness used in $\pi$ by $R$.
The set of randomness that may make mistake on some input $x \in \Bset^n$ is then defined as
\begin{equation*}
R_\mathsf{wrong}
:= \{r \in R \mid \exists x \st \pi_r(x) \neq f(x)\}
\end{equation*}
where $\pi_r(x)$ denotes the output of the algorithm $\pi$ when the input is $x \in \Bset^n$ and the randomness is $r \in R$.
\par
Since $\pi$ has the worst-case error $\leq \varepsilon$, for any $x \in \Bset^n$, 
\begin{equation*}
\Pr_R(\{r \in R \mid \pi_r(x) \neq f(x)\}) < 2^{-n}.
\end{equation*}
By summing up all over $x \in \Bset^n$, this leads to 
\begin{equation*}
\Pr_R(R_\mathsf{wrong}) \leq \sum_{x \in \Bset^n} \Pr_R(\{r \in R \mid \pi_r(x) \neq f(x)\}) < 1.
\end{equation*}
This means that there exists $r_\mathsf{good} \in R\setminus R_\mathsf{wrong}$, which satisfies
$\pi_{r_\mathsf{good}}(x) = f(x)$ for any $x \in \Bset^n$.
Fixing the randomness to $r_\mathsf{good}$, the deterministic algorithm $\pi_{r_\mathsf{good}}$ always output the correct value.
This means $R([f,0]) \leq |\pi_{r_\mathsf{good}}|$.
Since the new algorithm $\pi_{r_\mathsf{good}}$ must have a smaller (or at most equal) complexity than $|\pi|$, 
we also observe $|\pi_{r_\mathsf{good}}| \leq R([f, \varepsilon])$.
Together with the trivial relation $R([f, \varepsilon]) \leq R([f, 0])$, these arguments shows the desired statement.
\end{proof}

\section{Minimax theorem for algorithms}
Here we prove a minimax theorem for oracle algorithms. Since both of the classical and quantum cases are proved in the same manner, we focus on the classical randomized case.

\begin{Prop}\label{Prop_minimax}

Let $\tilde{\mathcal{U}} \subset \mathcal{P}(\Theta)$ is a non-empty, convex and compact subset, and $F: \FE \to \mathbb{R}$ be a function satisfying the following:

\begin{itemize}
\item For any finite distribution $\nu_A$ on $\FE$ and any $\mu \on \TU$,  $F(\pi(\nu_A), \mu) \leq \E_{\pi \sim \nu_A}[F(\pi, \mu)] $ where $\pi(\nu_A)$ is a randomized algorithm according to $\nu_A$.
\item For any finite distribution $\nu_B$ on $\tilde{\mathcal{U}}$ and any $\pi \in \FE$, $\E_{\nu_B}[F(\pi, \mu)] \leq F(\pi, \bar{\mu})$ where $\bar{\mu} := \E_{\nu_B}[\mu]$.
\item $F(\pi, \mu)$ is continuous with respect to $\mu \in \tilde{\mathcal{U}}$.
\end{itemize}
Then
\begin{equation*}
\inf_{\pi \in \FE} \max_{\mu \in \tilde{\mathcal{U}}}F(\pi, \mu)
=
 \max_{\mu \in \tilde{\mathcal{U}}}\inf_{\pi \in \FE}F(\pi, \mu)
\end{equation*}
\end{Prop}
The following lemma is a key ingredient for proof of Proposition~\ref{Prop_minimax}.

\begin{Lem}\label{Lem_finite_minimax}
For any finite subset $H_\pi \subset \FE$, any finite subset $H_\mu \subset \TU$ and any $\alpha \in \{\alpha' \geq \max_{\mu \in \TU} \min_{\pi \in H_\pi} F(\pi, \mu)\}$,
\begin{equation*}
\min_{\nu_A \text{~on~} H_\pi}\max_{\nu_B \text{~on~} H_\mu} \E_{\nu_A, \nu_B}[F(\pi, \mu)] \leq \alpha.
\end{equation*}
\end{Lem}
\begin{proof}
For any finite subset $H_\pi \subset \FE$ and any finite subset $H_\mu \subset \TU$, we first show
\begin{equation*}
\forall~\nu_B:\text{finite distribution on $\TU$}, \quad \exists\tau \in H_\pi \text{~s.t.~} \E_{\mu \sim \nu_B}[F(\tau, \mu)] \leq \alpha.
\end{equation*}
For any $\nu_B$, define $\bar{\mu} := \E_{\mu \sim \nu_B}[\mu(x)]$.
Then, for any $\alpha \in \{\alpha' \geq \max_{\mu \in \TU} \min_{\pi \in H_\pi} F(\pi, \mu)\}$, there is an algorithm $\tau \in H_\pi$ such that $F(\tau, \bar{\mu}) \leq \alpha$.
Therefore, by the convexity $\E_{\mu \sim \nu_B}[F(\tau, \mu)] \leq F(\tau, \bar{\mu})$, we obtain
\begin{equation*}
\forall~\nu_B:\text{finite distribution on $\TU$}, \quad \exists\tau \in H_\pi \text{~s.t.~} \E_{\mu \sim \nu_B}[F(\tau, \mu)] \leq F(\tau, \bar{\mu})\leq \alpha
\end{equation*}
which leads to 
\begin{equation*}
\max_{\nu_B \text{~on~} H_\mu} \min_{\nu_A \text{~on~} H_\pi}\E_{\nu_A, \nu_B}[F(\pi, \mu)] \leq \alpha.
\end{equation*}
Let us apply von-Neumann's minimax theorem here. 
\begin{equation*}
\min_{\nu_A \text{~on~} H_\pi}\max_{\nu_B \text{~on~} H_\mu} \E_{\nu_A, \nu_B}[F(\pi, \mu)] \leq \alpha
\end{equation*}
which completes proof.
\end{proof}

\begin{Lem}\label{Lem_inf_minimax}
For any finite subset $H_\pi \subset \FE$ and any $\alpha \in \{\alpha' \geq \max_{\mu \in \TU} \min_{\pi \in H_\pi} F(\pi, \mu)\}$, there is $\tau \in \FE$ such that
\begin{equation*}
\max_{\mu \in \TU} F(\tau, \mu) \leq \alpha.
\end{equation*}
\end{Lem}
\begin{proof}
We first show that 
for any $\varepsilon > 0$, 
\begin{equation}\label{eq_Lem_inf_minimax}
\min_{\nu_A \on H_\pi} \max_{\mu \in \TU} \E_{\nu_A}[F(\pi, \mu)] < \alpha + \varepsilon.
\end{equation}
For any $\pi \in H_\pi$, $F(\pi, \mu)$ is continuous on the compact set $\TU$. This means $F(\pi, \mu)$ is uniformly continuous, and therefore, 
for any $\varepsilon > 0$, there is  $\delta > 0$ such that
\begin{equation*}
\|\mu_1 - \mu_2\| < \delta \Rightarrow  \forall \pi \in H_\pi, \quad |F(\pi, \mu_1) - F(\pi, \mu_2)| < \varepsilon
\end{equation*}
holds. Note that $H_\pi$ is finite.
The compactness of the set $\TU$ also ensures that there is a finite set $\{\mu_1, \ldots, \mu_n\} \subset \TU$ such that 
\begin{equation*}
\bigcup_{i \leq n} B(\mu_i, \delta) 
= \TU.
\end{equation*}
Define $H_\mu(\varepsilon) := \{\mu_1, \ldots, \mu_n\}$. (Note that $H_\mu(\varepsilon)$ directly depends on $\delta$, and $\delta$ actually depends on $\varepsilon$. This means $H_\mu$ is in fact a function of $\varepsilon$.)
Now by Lemma~\ref{Lem_finite_minimax}, for any $\varepsilon >0$, 
\begin{equation*}
\min_{\nu_A \text{~on~} H_\pi}\max_{\nu_B \text{~on~} H_\mu(\varepsilon)} \E_{\nu_A, \nu_B}[F(\pi, \mu)] \leq \alpha
\end{equation*}
holds. This implies that there is  $\nu_A$ on  $H_\pi$ such that 
\begin{equation*}
\forall \nu_B \on H_\mu(\varepsilon), \quad \E_{\nu_A, \nu_B}[F(\pi, \mu)] \leq \alpha.
\end{equation*}
This leads to, 
\begin{equation*}
\forall \mu_i \in H_\mu(\varepsilon), \quad \E_{\pi \sim \nu_A}[F(\pi, \mu_i)]\leq \alpha.
\end{equation*}
Therefore, by the definition of $H_\mu(\varepsilon)$, we obtain
\begin{equation*}
\forall \mu \in \TU,\quad \E_{\pi \sim \nu_A}[F(\pi, \mu)] < \alpha + \varepsilon
\end{equation*}
which implies
\begin{equation*}
\min_{\nu_A \on H_\pi}\max_{\mu \in \TU} \E_{\nu_A}[F(\pi, \mu)] < \alpha + \varepsilon
\end{equation*}
and therefore the statement~\eqref{eq_Lem_inf_minimax} holds.
\par
Since the statement~\eqref{eq_Lem_inf_minimax} implies
$\min_{\nu_A \on H_\pi}\max_{\mu \in \TU} \E_{\nu_A}[F(\pi, \mu)] < \alpha$, we can take a distribution $\nu_A^0 \in H_\pi$ such that 
$\max_{\mu \in \TU} \E_{\nu_A^0}[F(\pi, \mu)] < \alpha$ holds. Therefore, together with the convexity $F(\pi(\nu_A), \mu) \leq \E_{\pi \sim\nu_A}[F(\pi, \mu)]$, 
we see that the randomized algorithm $\tau := \pi(\nu_A^0)$ satisfies
$\max_{\mu \in \TU} F(\tau, \mu) < \alpha$.
This completes proof.
\end{proof}
Using Lemma~\ref{Lem_finite_minimax} and Lemma~\ref{Lem_inf_minimax}, we show Proposition~\ref{Prop_minimax} as follows.
\begin{proof}[Proof of Proposition~\ref{Prop_minimax}]
Choose any $\alpha > \max_{\mu \in \TU} \inf_{\pi \in \FE} F(\pi, \mu)$ and define $A(\pi) := \{\mu \in \TU \mid F(\pi, \mu) \geq \alpha\}$.
Then $\bigcap_{\pi \in \FE} A(\pi) = \emptyset$. Since $\TU$ is compact and $A(\pi)$ is closed due to the continuity of $F(\pi, \mu)$, 
we see there is a finite set of algorithms $H_\pi \subset \FE$ such that $\bigcap_{\pi \in \FE} A(\pi) = \emptyset$.
Thus we have that 
$\forall \mu \in \TU, \quad \min_{\pi \in H_\pi} F(\pi, \mu)< \alpha$
which is equivalent to 

\begin{equation*}
\max_{\mu \in \TU}\min_{\pi \in H_\pi} F(\pi, \mu)< \alpha.
\end{equation*}
Then Lemma~\ref{Lem_inf_minimax} tells that there is a algorithm $\tau \in \FE$ such that $\forall \mu \in \TU$, $F(\tau, \mu) \leq \alpha$.
To summarize, for any $\alpha > \max_{\mu \in \TU} \inf_{\pi \in \FE} F(\pi, \mu)$, there is a algorithm $\tau \in \FE$ such that $\min_{\mu \in \TU}F(\tau, \mu) \leq \alpha$.
This shows
\begin{equation*}
\inf_{\pi \in \FE}\max_{\mu \in \TU} F(\pi, \mu)\leq \max_{\mu \in \TU} \inf_{\pi \in \FE}F(\pi, \mu).
\end{equation*}
This completes proof.
\end{proof}

\end{document}